\documentclass[a4paper, USenglish, cleveref, autoref, thm-restate, pagebackref]{lipics-v2019}

\title{The Complexity of Binary Matrix Completion Under Diameter Constraints} 

\author{Tomohiro Koana}{Technische Universit\"at Berlin, Faculty~IV, Algorithmics and Computational Complexity}{tomohiro.koana@tu-berlin.de}{[orcid]}{Partially supported by the DFG project MATE (NI 369/17).} 

\author{Vincent Froese}{Technische Universit\"at Berlin, Faculty~IV, Algorithmics and Computational Complexity}{vincent.froese@tu-berlin.de}{[orcid]}{}

\author{Rolf Niedermeier}{Technische Universit\"at Berlin, Faculty~IV, Algorithmics and Computational Complexity}{rolf.niedermeier@tu-berlin.de}{[orcid]}{}

\authorrunning{T.~Koana, V.~Froese, and R.~Niedermeier} 

\Copyright{Tomohiro Koana, Vincent Froese, and Rolf Niedermeier} 



\acknowledgements{We are grateful to Christian Komusiewicz for helpful feedback on an earlier version of this work and to Stefan Szeider for pointing us to the work on clustering incomplete data~\cite{EGKOS21}.
We also thank Curtis Bright for mentioning the connection to the Hadamard matrix completion problem.}

\nolinenumbers 

\hideLIPIcs  

\EventEditors{}
\EventNoEds{0}
\EventLongTitle{}
\EventShortTitle{}
\EventAcronym{}
\EventYear{}
\EventDate{}
\EventLocation{}
\EventLogo{}
\SeriesVolume{}
\ArticleNo{}
\usepackage{amsmath}
\usepackage{amsfonts}
\usepackage{amssymb}
\usepackage{thm-restate}

\usepackage{tabularx,array,booktabs,multirow}
\usepackage{enumitem}

\usepackage{hyperref}
\usepackage{cleveref}
\usepackage{url}

\usepackage{graphicx}
\usepackage{caption}
\usepackage{subcaption}
\usepackage{tikz}
\usetikzlibrary{arrows,decorations.pathmorphing,decorations.pathreplacing,backgrounds,positioning,fit,matrix}
\usetikzlibrary{shapes,calc}

\usepackage{xcolor}
\definecolor{lipicsLightGray}{gray}{0.85}

\usepackage[ruled]{algorithm}
\usepackage[noend]{algpseudocode}
\algnewcommand\algorithmicinput{\textbf{Input:}}
\algnewcommand\algorithmicoutput{\textbf{Task:}}
\algnewcommand\Input{\item[\algorithmicinput]}
\algnewcommand\Output{\item[\algorithmicoutput]}

\usepackage{todonotes}
\usepackage{xspace}

\newcommand{\NN}{\mathbb{N}} 
\newcommand{\dist}{d}
\newcommand{\mind}{\gamma}
\newcommand{\maxd}{\delta}

\DeclareMathOperator{\core}{core}

\newcommand{\mc}{{\normalfont\textsc{DMC}}}

\newcommand{\problemdef}[4]{
\begin{center}
  \begin{minipage}{0.95\textwidth}
    \normalsize\textsc{#2} \smallskip \\
    \begin{tabularx}{\textwidth}{@{}l@{\hspace{3pt}}X}
      \normalsize\textbf{Input:} & \normalsize#3 \\
      \normalsize\textbf{#1:}    & \normalsize#4
    \end{tabularx}
  \end{minipage}
\end{center}
}
\newcommand{\dprob}[4][Question]{\problemdef{#1}{#2}{#3}{#4}}

\newtheorem{obs}[theorem]{Observation}
\crefname{obs}{observation}{Observation}

\begin{document}

\maketitle

\begin{abstract}
  We thoroughly study a novel but basic combinatorial matrix completion problem: Given a binary incomplete matrix, fill in the missing entries so that every pair of rows in the resulting matrix has a Hamming distance within a specified range.
  We obtain an almost complete picture of the complexity landscape regarding the distance constraints and the maximum number of missing entries in any row.
  We develop polynomial-time algorithms for maximum diameter three based on Deza’s theorem [Discret.~Math.~1973] from extremal set theory.
  We also prove NP-hardness for diameter at least four.
  For the number of missing entries per row, we show polynomial-time solvability when there is only one and NP-hardness when there can be at least two. In many of our algorithms, we  heavily rely on Deza’s theorem to identify sunflower structures. This paves the way towards polynomial-time algorithms which are based on finding graph factors and solving 2-SAT instances.
\end{abstract}

\section{Introduction}
In combinatorial matrix completion problems, given an incomplete 
matrix over a fixed 
alphabet with some missing entries, 
the goal is to fill in the missing entries such that the resulting 
``completed matrix'' (over the same alphabet) fulfills a desired property.
Performing a parameterized complexity analysis, 
Ganian et al.~\cite{GKOS20,GKOS18} and Eiben et al.~\cite{EGKOS21,EGKOS22} recently 
contributed to this growing field by studying various desirable 
properties. More specifically, Ganian et al.~\cite{GKOS18} studied the two
properties of minimizing the rank or of minimizing the number of 
distinct rows of the completed matrix. Ganian et al.~\cite{GKOS20} analyzed the
complexity of completing an incomplete matrix so that it fulfills certain 
constraints and can be partitioned into subspaces of small rank. 
Eiben et al.~\cite{EGKOS21} investigated 
clustering problems where one wants to partition the
rows of the completed matrix into a given number of clusters of small radius 
or of small diameter. 
Finally, Koana et al.~\cite{KFN20} studied two cases of completing the 
matrix into one which has small (local) radius. 
The latter two papers~\cite{EGKOS21,KFN20} rely on Hamming distance as a 
distance measure;  
in general, all considered matrix completion 
problems are NP-hard and thus the above 
papers~\cite{EGKOS21,EGKOS22,GKOS20,GKOS18,KFN20} mostly focused on parameterized complexity
studies. In this work, we focus on a desirable property 
closely related to small radius, namely diameter bounds.
Doing so, we further focus on the case of binary alphabet.
For a matrix~$\mathbf{T}\in\{0,1\}^{n\times \ell}$, let $\mind(\mathbf{T}) := \min_{i \ne i' \in [n]} \dist(\mathbf{T}[i], \mathbf{T}[i'])$ and $\maxd(\mathbf{T}) := \max_{i \ne i' \in [n]} \dist(\mathbf{T}[i], \mathbf{T}[i'])$,
where $\dist$ denotes the Hamming distance and $\mathbf{T}[i]$ denotes the $i$-th row of $\mathbf{T}$.
We use the special symbol~$\square$ to represent a missing entry.
Specifically, we study the following problem.

\dprob
{Diameter Matrix Completion (DMC)}
{An incomplete matrix $\mathbf{S} \in \{ 0, 1, \square \}^{n \times \ell}$ and $\alpha \le \beta \in \NN$.}
{Is there a completion $\mathbf{T} \in \{ 0, 1 \}^{n \times \ell}$ of $\mathbf{S}$ with $\alpha \le \gamma(\mathbf{T})$ and $\maxd(\mathbf{T}) \le \beta$?}

We believe that \mc{} is a natural combinatorial matrix problem which may appear 
in the following contexts:
\begin{itemize}
\item In coding theory, one may want to ``design'' (by filling in the 
missing entries) codewords that are pairwise neither too close 
(parameter~$\alpha$ in DMC) nor too far (parameter~$\beta$ in
DMC) from each other.
One prime example is the completion into a Hadamard matrix~\cite{Joh90}.
This is a special case of \textsc{DMC} with $n = \ell$ and $\alpha = \beta = n / 2$.
\item In computational biology, one may want to minimize the maximum distance
of sequences in order to determine their degree of relatedness
(thus minimizing~$\beta$); missing entries refer to missing data 
points.\footnote{Here, it would be particularly natural to also study the case of non-binary alphabets;
however, most of our positive results probably only hold for binary alphabet.}
\item In data science, each row may represent an entity with its attributes, 
	and 
solving the DMC problem may fulfill some constraints with respect to
		the pairwise (dis)similarity of the completed entities.
\item In stringology, DMC seems to constitute a new and natural 
problem, closely related to several intensively studied consensus 
problems (many of which are NP-hard for binary alphabets)~\cite{BLL15,BLR05,BFN20,BHKN14,BP0R018,BS20,CLPPS16,GNR03,HR15,LKLB14,LSLI02,Sch17,WLM16}.
\end{itemize}
Somewhat surprisingly, although simple to define and well-motivated, 
in the literature there
seems to be no systematic study of DMC and its computational complexity.

\begin{figure}[t]
  \centering
  \begin{subfigure}{.45\textwidth}
    \centering
    \begin{tikzpicture}[scale=0.42]

      \pgfmathsetmacro{\xmax}{9}
      \pgfmathsetmacro{\xmaxd}{\xmax - 1}
      \pgfmathsetmacro{\xmaxdd}{\xmax - 2}
      \pgfmathsetmacro{\xmaxddd}{\xmax - 3}
      \pgfmathsetmacro{\ymax}{9}
      \pgfmathsetmacro{\ymaxd}{\ymax - 1}


      \foreach \i in { 0, ..., \xmaxd } \node at (\i + .5, -.5) {\i};
      \foreach \i in { 0, ..., \ymaxd } \node at (-.5, \i + .5) {\i};

      \fill[red!35] (0, 4) -- (1, 4) -- (1, 6) -- (3, 6) -- (3, 8) -- (5, 8) -- (5, 9) -- (0, 9) -- (0, 4);

      \tikzstyle{greenfill}=[green!35]
      \fill[greenfill] (0, 0) -- (0, 4) -- (1, 4) -- (1, 0) -- (0, 0);
      \foreach \i in { 1, ..., \xmaxdd } {
        \fill[greenfill] (\i, \i) -- (\i, \i + 2) -- (\i + 1, \i + 2) -- (\i + 1, \i) -- (\i, \i);
      }
      \fill[greenfill] (\xmaxd, \xmaxd) -- (\xmaxd + 1, \xmaxd) -- (\xmaxd + 1, \xmaxd + 1) -- (\xmaxd, \xmaxd + 1) -- (\xmaxd, \xmaxd);

      \draw[red!70,very thick] (0, 4) -- (1, 4) -- (1, 6) -- (3, 6) -- (3, 8) -- (5, 8) -- (5, 9);
      \tikzstyle{greenline}=[green!80, very thick]
      \foreach \i in { 1, ..., 6 } {
        \draw[greenline] (\i, \i + 2) -- (\i + 1, \i + 2);
      }
      \foreach \i in { 1, ..., 8 } {
        \draw[greenline] (\i, \i) -- (\i + 1, \i);
      }
      \foreach \i in { 0, ..., 5 } {
        \draw[greenline] (\i + 2, \i + 3) -- (\i + 2, \i + 4);
      }
      \foreach \i in { 0, ..., 8 } {
        \draw[greenline] (\i + 1, \i) -- (\i + 1, \i + 1);
      }
      \draw[greenline] (1, 3) -- (1, 4) -- (0, 4);
      \draw[line width=.69pt,xshift=1.4pt,yshift=.6pt,red!70] (0, 4) -- (1, 4);

      \foreach \i in { 1, ..., \xmax } {
        \draw[help lines] (\i, \i - 1) -- (\i, \ymax);
      }
      \foreach \i in { 1, ..., \ymax } {
        \draw[help lines] (0, \i - 1) -- (\i, \i - 1);
      }
      \tikzstyle{axis}=[->, very thick, >=stealth']
      \draw[axis] (0, 0) -- (\xmax + .3, 0) node [right] {$\alpha$};
      \draw[axis] (0, 0) -- (0, \ymax + .3) node [above] {$\beta$};
    \end{tikzpicture}
    \caption{Complexity of \mc{} with respect to combinations of constant values for $\alpha$ and~$\beta$.}
    \label{fig:const}
  \end{subfigure}
  \quad\quad
  \begin{subfigure}{.45\textwidth}
    \centering
    \begin{tikzpicture}[scale=0.42]
      \pgfmathsetmacro{\xmax}{9}
      \pgfmathsetmacro{\xmaxd}{\xmax - 1}
      \pgfmathsetmacro{\xmaxdd}{\xmax - 2}
      \pgfmathsetmacro{\xmaxddd}{\xmax - 3}
      \pgfmathsetmacro{\ymax}{9}
      \pgfmathsetmacro{\ymaxd}{\ymax - 1}
      \tikzstyle{redline}=[red!75, very thick]
      \fill[red!35] (2, 3) -- (\xmax, 3) -- (\xmax, \ymax) -- (2, \ymax) -- (2, 3);
      \fill[red!35] (3, 0) -- (\xmax, 0) -- (\xmax, 1) -- (3, 1) -- (3, 0);
      \draw[redline] (2, 3) -- (\xmax, 3);
      \draw[redline] (3, 1) -- (\xmax, 1);

      \tikzstyle{greenfill}=[green!35]
      \fill[greenfill] (0, 0) -- (0, \ymax) -- (2, \ymax) -- (2, 1) -- (3, 1) -- (3, 0) -- (0, 0);
      \tikzstyle{greenline}=[green!80, very thick]
      \draw[greenline] (2, \ymax) --  (2, 1) -- (3, 1) -- (3, 0);
      \draw[line width=.7pt,xshift=.6pt,yshift=-1.4pt,red!70] (2, 3) -- (2, 9.05);
      \draw[line width=.7pt,xshift=.6pt,red!70] (3, 0) -- (3, 1.05);

      \foreach \i in { 1, ..., \xmaxd } {
        \draw[help lines] (\i, 0) -- (\i, \ymax);
      }
      \foreach \i in { 1, ..., \ymaxd } {
        \draw[help lines] (0, \i) -- (\xmax, \i);
      }
      \tikzstyle{axis}=[->, very thick, >=stealth']
      \draw[axis] (0, 0) -- (\xmax + .3, 0) node [right] {$k$};
      \draw[axis] (0, 0) -- (0, \ymax + .3) node [above] {$\beta - \alpha$};

      \foreach \i in { 0, ..., \xmaxd } \node at (\i + .5, -.5) {\i};
      \foreach \i in { 0, ..., \ymaxd } \node at (-.5, \i + .5) {\i};
    \end{tikzpicture}
    \caption{Complexity of \mc{} with respect to combinations of the maximum number~$k$ of missing entries in any row and $\beta-\alpha$.}
    \label{fig:diff}
  \end{subfigure}
  \caption{
	  Overview of our results. Green (lighter) denotes polynomial-time solvability and red (darker) denotes NP-hardness. White cells indicate open cases.
}
\label{fig:results}
\end{figure}

We perform a fine-grained complexity study in terms of diameter bounds $\alpha$, $\beta$ and the maximum number $k$ of missing entries in any row.
Note that in bioinformatics applications matrix rows may represent sequences with few corrupted data points, thus resulting in small values for~$k$.
In fact, the computational complexity with respect to this kind of parameters has been studied in the context of computational biology \cite{BLR05,BHKN14,HR15}.
We identify polynomial-time cases as well as NP-hard cases, taking significant steps towards a computational complexity 
dichotomy (polynomial-time solvable versus NP-hard), leaving fairly few 
cases open.  While the focus of the previous works~\cite{EGKOS21,EGKOS22,KFN20} is 
on parameterized complexity studies, in this work we settle
more basic algorithmic questions on the DMC~problem, relying on several 
combinatorial insights,  including results from (extremal) combinatorics 
(most prominently, Deza's theorem~\cite{Deza1974solution}).
Indeed, we believe that exploiting sunflowers based on Deza's theorem 
in combination with the corresponding use of algorithms for 2-SAT and graph factors
 is our most interesting technical contribution.
In this context, we also observe the phenomenon that the running time 
bounds that we can prove for odd values of~$\alpha$ (the ``lower bound for 
dissimilarity'') are significantly better than the ones for even values
of~$\alpha$---indeed, for even values of~$\alpha$ the running time 
exponentially depends on~$\alpha$ while it is independent of~$\alpha$
for odd values of~$\alpha$.
We survey our results in Figure~\ref{fig:results} which also depicts 
remaining open cases. 

\paragraph{Related work}
The closest studies are the work of Hermelin and Rozenberg~\cite{HR15}, Koana et al.~\cite{KFN20}, and Eiben et al.~\cite{EGKOS21,EGKOS22}.
Hermelin and Rozenberg~\cite{HR15} and Koana et al.~\cite{KFN20} studied \textsc{Constraint Radius Matrix Completion}:

\dprob
{Constraint Radius Matrix Completion (ConRMC)}
{An incomplete matrix $\mathbf{S} \in (\Sigma \cup \square)^{n \times \ell}$ and $r \in \NN^n$.}
{Is there a completion $\mathbf{T} \in \Sigma^{n \times \ell}$ of $\mathbf{S}$ and a row vector $v \in \Sigma^\ell$ such that $\dist(v, \mathbf{T}[i]) \le r[i]$ for all $i \in [n]$?}

Note that \textsc{ConRMC} is defined for arbitrary alphabets. 
A more important difference between DMC and \textsc{ConRMC} 
is that in DMC we basically have to compare all rows against each other,
but in \textsc{ConRMC} we have to compare one ``center row'' against all 
others. Indeed, this makes these two similarly defined problems quite 
different in many computational complexity aspects.
See the  example in \Cref{fig:dmc}
that also illustrates significant differences between radius 
minimization and diameter minimization (the latter referring to
$\maxd(\mathbf{T}) \le \beta$ above).
Recall that $k$ is the  maximum number of missing entries in any row.
The special case $k = 0$ is a generalization of the well-known \textsc{Closest String} problem, which is NP-hard \cite{FL97}. 
Hermelin and Rozenberg~\cite{HR15} proved that \textsc{ConRMC} is NP-hard even if $\max_{i \in [n]} r[i] = 2$ while it is polynomial-time solvable for $\max_{i \in [n]} r[i] = 1$.
Koana et al.~\cite{KFN20} provided a linear-time algorithm for $\max_{i \in [n]} r[i] = 1$.
Moreover, Koana et al.~\cite{KFN20} established fixed-parameter tractability with respect to $\max_{i \in [n]} r[i] + k$.

Eiben et al.~\cite{EGKOS21} studied the following problem among others:
\dprob
{Diameter Clustering Completion (DCC)}
{An incomplete matrix $\mathbf{S} \in \{ 0, 1, \square \}^{n \times \ell}$ and $r, c \in \NN$.}
{Is there a completion $\mathbf{T} \in \{ 0, 1 \}^{n \times \ell}$ of $\mathbf{S}$ and a partition $(I_1, \dots, I_c)$ of $[n]$ such that $\maxd(\mathbf{T}[I_i]) \le r$ for every $i \in [n]$?}

Here, for $I \subseteq [n]$, $\mathbf{T}[I]$ is the submatrix comprising the rows with index $i \in I$.
\mc{} and \textsc{DCC} are closely related.
In fact, \mc{} with $\alpha=0$ is equivalent to \textsc{DCC} with $c = 1$.
However, the problems are incomparable:
While \mc{} also models the aspect of achieving a minimum pairwise distance (not only a bounded diameter), \textsc{DCC} focuses on clustering.
Eiben et al.~\cite{EGKOS21} showed NP-hardness for $c \ge 3$.
They also showed that \textsc{DCC} on complete matrices is NP-hard for $r = 6$.
Furthermore, using the classical sunflower lemma \cite{ER60}, they proved fixed-parameter tractability with respect to $r + c + \textsf{cover}$, where $\textsf{cover}$ is the minimum number of rows and columns whose deletion results in a matrix without any missing entries.
We remark that this parameter $\textsf{cover}$ is not comparable to the maximum number $k$ of missing entries in any row.
To see this, consider the following two square matrices $\mathbf{M}_1, \mathbf{M}_2 \in \{ 0, 1, \square \}^{n \times n}$, where an entry is missing in $\mathbf{M}_1$ if and only if it is on the main diagonal and an entry is missing in $\mathbf{M}_2$ if and only if it is on the last row.
Observe that $k = 1$ and $\textsf{cover} = n$ for $\mathbf{M}_1$ and that $k = n$ and $\textsf{cover} = 1$ for $\mathbf{M}_2$.

There are also numerous work on non-combinatorial matrix completion problems in the context of clustering \cite{DBLP:journals/dcg/GaoLS08,DBLP:journals/talg/GaoLS10,DBLP:conf/soda/LeeS13,DBLP:conf/nips/MaromF19,DBLP:conf/soda/EibenFGLPS21} such as $k$-center clustering and $k$-means clustering.
These clustering problems are NP-hard even if the matrix has no missing entry. Often, the focus here is on developing approximation algorithms for the corresponding matrix completion problems.

\begin{figure}[t]
  \centering
  \begin{tikzpicture}[scale=0.58, every node/.style={font=\footnotesize}]
    \fill [lipicsLightGray] (0, 1) rectangle (1, 2);
    \fill [lipicsLightGray] (1, 2) rectangle (2, 3);
    \fill [lipicsLightGray] (1, 3) rectangle (2, 4);
    \fill [lipicsLightGray] (2, 0) rectangle (3, 1);
    \fill [lipicsLightGray] (2, 3) rectangle (3, 4);
    \fill [lipicsLightGray] (3, 1) rectangle (4, 2);
    \fill [lipicsLightGray] (3, 2) rectangle (4, 3);
    \fill [lipicsLightGray] (4, 0) rectangle (5, 1);
    \fill [lipicsLightGray] (4, 3) rectangle (5, 4);

    \node at (0.5, 0.5) {$0$};
    \node at (0.5, 1.5) {$1$};

    \node at (1.5, 1.5) {$0$};
    \node at (1.5, 2.5) {$1$};
    \node at (1.5, 3.5) {$1$};

    \node at (2.5, 0.5) {$1$};
    \node at (2.5, 1.5) {$0$};
    \node at (2.5, 2.5) {$0$};
    \node at (2.5, 3.5) {$1$};

    \node at (3.5, 0.5) {$0$};
    \node at (3.5, 1.5) {$1$};
    \node at (3.5, 2.5) {$1$};
    \node at (3.5, 3.5) {$0$};

    \node at (4.5, 0.5) {$1$};
    \node at (4.5, 1.5) {$0$};
    \node at (4.5, 2.5) {$0$};
    \node at (4.5, 3.5) {$1$};
    \draw[help lines] (0,0) grid (5, 4);

    \draw[very thick] (0, 2) rectangle (1, 3);
    \draw[very thick] (0, 3) rectangle (1, 4);
    \draw[very thick] (1, 0) rectangle (2, 1);
    \begin{scope}[shift={(6.5, 0)}]
      \fill [lipicsLightGray] (0, 1) rectangle (1, 2);
      \fill [lipicsLightGray] (0, 3) rectangle (1, 4);
      \fill [lipicsLightGray] (1, 2) rectangle (2, 3);
      \fill [lipicsLightGray] (1, 3) rectangle (2, 4);
      \fill [lipicsLightGray] (2, 0) rectangle (3, 1);
      \fill [lipicsLightGray] (2, 3) rectangle (3, 4);
      \fill [lipicsLightGray] (3, 1) rectangle (4, 2);
      \fill [lipicsLightGray] (3, 2) rectangle (4, 3);
      \fill [lipicsLightGray] (4, 0) rectangle (5, 1);
      \fill [lipicsLightGray] (4, 3) rectangle (5, 4);
      \node at (0.5, 0.5) {$0$};
      \node at (0.5, 1.5) {$1$};
      \node at (0.5, 2.5) {$0$};
      \node at (0.5, 3.5) {$1$};

      \node at (1.5, 0.5) {$0$};
      \node at (1.5, 1.5) {$0$};
      \node at (1.5, 2.5) {$1$};
      \node at (1.5, 3.5) {$1$};

      \node at (2.5, 0.5) {$1$};
      \node at (2.5, 1.5) {$0$};
      \node at (2.5, 2.5) {$0$};
      \node at (2.5, 3.5) {$1$};

      \node at (3.5, 0.5) {$0$};
      \node at (3.5, 1.5) {$1$};
      \node at (3.5, 2.5) {$1$};
      \node at (3.5, 3.5) {$0$};

      \node at (4.5, 0.5) {$1$};
      \node at (4.5, 1.5) {$0$};
      \node at (4.5, 2.5) {$0$};
      \node at (4.5, 3.5) {$1$};
      \draw[help lines] (0,0) grid (5, 4);
      \draw[very thick] (0, 2) rectangle (1, 3);
      \draw[very thick] (0, 3) rectangle (1, 4);
      \draw[very thick] (1, 0) rectangle (2, 1);
    \end{scope}
    \begin{scope}[shift={(13, .7)}]
      \fill [lipicsLightGray] (0, 1) rectangle (1, 2);
      \fill [lipicsLightGray] (1, 2) rectangle (2, 3);
      \fill [lipicsLightGray] (1, 3) rectangle (2, 4);
      \fill [lipicsLightGray] (2, 0) rectangle (3, 1);
      \fill [lipicsLightGray] (2, 3) rectangle (3, 4);
      \fill [lipicsLightGray] (3, 1) rectangle (4, 2);
      \fill [lipicsLightGray] (3, 2) rectangle (4, 3);
      \fill [lipicsLightGray] (4, 0) rectangle (5, 1);
      \fill [lipicsLightGray] (4, 3) rectangle (5, 4);
      \node at (0.5, 0.5) {$0$};
      \node at (0.5, 1.5) {$1$};
      \node at (0.5, 2.5) {$0$};
      \node at (0.5, 3.5) {$0$};

      \node at (1.5, 0.5) {$0$};
      \node at (1.5, 1.5) {$0$};
      \node at (1.5, 2.5) {$1$};
      \node at (1.5, 3.5) {$1$};

      \node at (2.5, 0.5) {$1$};
      \node at (2.5, 1.5) {$0$};
      \node at (2.5, 2.5) {$0$};
      \node at (2.5, 3.5) {$1$};

      \node at (3.5, 0.5) {$0$};
      \node at (3.5, 1.5) {$1$};
      \node at (3.5, 2.5) {$1$};
      \node at (3.5, 3.5) {$0$};

      \node at (4.5, 0.5) {$1$};
      \node at (4.5, 1.5) {$0$};
      \node at (4.5, 2.5) {$0$};
      \node at (4.5, 3.5) {$1$};
      \draw[help lines] (0,0) grid (5, 4);
      \draw[very thick] (0, 2) rectangle (1, 3);
      \draw[very thick] (0, 3) rectangle (1, 4);
      \draw[very thick] (1, 0) rectangle (2, 1);

      \begin{scope}[shift={(0,-1.5)}]
        \node at (0.5, 0.5) {$0$}; \node at (1.5, 0.5) {$0$}; \node at (2.5, 0.5) {$0$}; \node at (3.5, 0.5) {$0$}; \node at (4.5, 0.5) {$0$};
        \draw[help lines] (0,0) grid (5, 1);
      \end{scope}
    \end{scope}
  \end{tikzpicture}
  \caption{An illustration of matrix completion problems with the input matrix (left).
  Missing entries (and their completions) are framed by thick lines.
  The middle matrix is a completion of diameter four and the right matrix is a completion of radius three with the center vector below.
  Note that missing entries in the same column might be filled with different values to meet the diameter constraint, whereas this is never necessary for the radius constraint.}
  \label{fig:dmc}
\end{figure}

\paragraph*{Organization}
We provide preliminaries in \Cref{sec:prelim}.
In \Cref{sec:mindmc:d}, we study \textsc{DMC} with fixed constants $\alpha \le \beta$ (see \Cref{fig:const}).
We start with polynomial-time solvability:
\Cref{sec:alpha0} and \Cref{sec:const:poly} cover the case $\alpha = 0$ and $\beta \le 3$ and the case $\alpha > 0$ and $\beta \in \{ \alpha, \alpha + 1 \}$, respectively.
We prove NP-hardness for $\beta \ge 2 \lceil \alpha / 2 \rceil + 4$ in \Cref{sec:const:nph}.
In \Cref{sec:mindmc:k2}, we consider \textsc{DMC} with unbounded $\alpha$ and $\beta$ (see \Cref{fig:diff}).
We prove polynomial-time solvability in \Cref{sec:unbounded:poly} and NP-hardness in \Cref{sec:unbounded:nph}.


\section{Preliminaries}
\label{sec:prelim}
For $m \le n \in \NN$, let $[m, n]:=\{ m, \dots, n \}$ and let $[n]:=[1, n]$.

For a matrix $\mathbf{T} \in \{ 0, 1 \}^{n \times \ell}$, we denote by $\mathbf{T}[i, j]$ the
entry in the~$i$-th row and~$j$-th column ($i \in [n]$ and $j \in [\ell]$) of~$\mathbf{T}$.
We use $\mathbf{T}[i, :]$ (or $\mathbf{T}[i]$ in short) to denote the \emph{row vector} $(\mathbf{T}[i, 1], \dots, \mathbf{T}[i, \ell])$ and $\mathbf{T}[:, j]$ to denote the \emph{column vector} $(\mathbf{T}[1, j], \dots, \mathbf{T}[n, j])^T$.
For subsets $I \subseteq [n]$ and $J \subseteq [\ell]$, we write $\mathbf{T}[I, J]$ to denote the submatrix containing only the rows in $I$ and the columns in $J$.
We abbreviate $\mathbf{T}[I, [\ell]]$ and $\mathbf{T}[[n], J]$ as $\mathbf{T}[I, :]$ (or $\mathbf{T}[I]$ for short) and $\mathbf{T}[:, J]$, respectively.
We use the special character~$\square$ for a \emph{missing} entry.
A matrix $\mathbf{S} \in \{ 0, 1, \square \}^{n \times \ell}$ is called \emph{incomplete} if it contains a missing entry, and it is called \emph{complete} otherwise.
We say that $\mathbf{T} \in \{ 0, 1 \}^{n \times \ell}$ is a \emph{completion} of $\mathbf{S} \in \{ 0, 1, \square \}^{n \times \ell}$ if either $\mathbf{S}[i, j] = \square$ or $\mathbf{S}[i, j] = \mathbf{T}[i, j]$ holds for all $i \in [n]$ and $j \in [\ell]$.

Let $u, w \in \{ 0, 1, \square \}^\ell$ be row vectors.
For $j \in [\ell]$ and $J \subseteq [\ell]$, let $u[j]$ denote the $j$-th entry of $u$ and let $u[J]$ denote the vector only containing entries in $J$.
Let~$D(u, w):=\{ j \in [\ell] \mid u[j] \ne w[j] \wedge u[j]\ne \square \wedge w[j]\ne \square \}$ be the set of column indices where~$u$ and~$v$ disagree (not considering positions with missing entries).
The \emph{Hamming distance} between $u$ and $w$ is $\dist(u, w) := |D(u, w)|$.
Note that the Hamming distance obeys the triangle inequality $d(u, w) \le d(u, v) + d(v, w)$ for a vector $v \in \{ 0, 1 \}^\ell$.
For a subset $J \subseteq [\ell]$, we also define $\dist_{J}(u, w):=\dist(u[J], w[J])$.
Let $u', v', w' \in \{ 0, 1 \}^\ell$ be complete row vectors.
Then, it holds that $\dist(u', w') = |D(u', v') \triangle D(v', w')| = |D(u', v')| + |D(v', w')| - 2|D(u', v') \cap D(v', w')|$.
The binary operation $u \leftarrow v$ replaces the missing entries of~$u$ with the corresponding entries in $v$ for $v \in \{ 0, 1 \}^\ell$.
We sometimes use string notation to represent row vectors, such as 001 for $(0, 0, 1)$.

\section{Constant Diameter Bounds \texorpdfstring{$\alpha$}{alpha} and \texorpdfstring{$\beta$}{beta}}
\label{sec:mindmc:d}

In this section, we consider the special case $(\alpha,\beta)$-\mc{} of \mc, where $\alpha\le
\beta$ are some fixed constants.
We prove the results depicted in \Cref{fig:const}.
To start with, we identify the following simple linear-time solvable case which will subsequently be used
several times.

\begin{lemma}
  \label{lemma:linearforconstantl}
  \mc{} can be solved in linear time for a constant number~$\ell$ of columns.
\end{lemma}
\begin{proof}
  If $\alpha > 0$ and $n > 2^\ell$, then there is no completion $\mathbf{T}$ of $\mathbf{S}$ with $\gamma(\mathbf{T}) \ge \alpha > 0$.
  Thus, we can assume that the input matrix comprises at most $n \ell \le 2^\ell \cdot \ell$ (that is, constantly many) entries for the case $\alpha > 0$.
  Suppose that $\alpha = 0$.
  Consider a set $\mathcal{V}\subseteq\{ 0, 1 \}^\ell$ in which the pairwise Hamming distances are at most $\beta$.
  We simply check whether each row vector in the input matrix can be completed to some row vector in $\mathcal{V}$ in~$O(n\cdot 2^\ell)=O(n)$ time.
  Since there are at most~$2^{2^{\ell}}$ choices for $\mathcal{V}$, this procedure can be done in linear time.
\end{proof}

\subsection{Polynomial time for \texorpdfstring{$\alpha = 0$}{alpha = 0} and \texorpdfstring{$\beta \le 3$}{beta = 2}}
\label{sec:alpha0}

As an entry point, we show that $(0, 1)$-\mc{} is easily solvable.
To this end, we call a column vector \emph{dirty} if it contains both~0 and~1.
Clearly, for~$\alpha=0$, we can ignore columns that are not dirty since they can always be completed without increasing the Hamming distances between rows.
Hence, throughout this subsection, we assume that the input matrix contains only dirty columns.
Now, any $(0, 1)$-\mc{} instance is a \textbf{Yes}-instance if and only if there is at most one dirty column in the input matrix:

\begin{lemma}
  \label{lemma:dmc:d1}
  A matrix $\mathbf{S} \in \{ 0, 1, \square \}^{n \times \ell}$ admits a completion $\mathbf{T} \in \{ 0, 1 \}^{n \times \ell}$ with $\delta(\mathbf{T}) \le 1$ if and only if $\mathbf{S}$ contains at most one dirty column.
\end{lemma}
\begin{proof}
  Suppose that \( \mathbf{S} \) contains two dirty columns $\mathbf{S}[:, j_0]$ and $\mathbf{S}[:, j_1]$ for $j_0 \ne j_1 \in [\ell]$.
  We claim that $\maxd(\mathbf{T}) \ge 2$ holds for any completion $\mathbf{T}$ of $\mathbf{S}$.
  Let $i \in [n]$.
  Then, there exist $i_0, i_1 \in [n]$ with $\mathbf{T}[i, j_0] \ne \mathbf{T}[i_0, j_0]$ and $\mathbf{T}[i, j_1] \ne \mathbf{T}[i_1, j_1]$.
  If $\maxd(\mathbf{T}) \le 1$, then we obtain $\mathbf{T}[i_0, j_1] = \mathbf{T}[i, j_1]$ and $\mathbf{T}[i_1, j_0] = \mathbf{T}[i, j_0]$.
  Now we have $\dist(\mathbf{T}[i_0], \mathbf{T}[i_1]) \ge 2$ because $\mathbf{T}[i_0, j_0] \ne \mathbf{T}[i_1, j_0]$ and $\mathbf{T}[i_0, j_1] \ne \mathbf{T}[i_1, j_1]$.
  The reverse direction follows easily.
\end{proof}
\Cref{lemma:dmc:d1} implies that one can solve $(0, 1)$-\mc{} in linear time.
In the following, we extend this to a linear-time algorithm for $(0, 2)$-\mc{} (\Cref{theorem:d2linear}) and a polynomial-time algorithm for $(0, 3)$-\mc{} (\Cref{theorem:d3polynomial}).

For these algorithms, we make use of a concept from extremal set theory, known as $\Delta$-systems~\cite{jukna2011extremal}.
We therefore consider matrices as certain set systems.

\begin{definition}
  \label{def:diffsets}
  For a matrix $\mathbf{T} \in \{ 0, 1 \}^{n \times \ell}$,
  let $\mathcal{T}$ denote the set system $\{ D(\mathbf{T}[i], \mathbf{T}[n]) \mid i \in [n - 1] \}$.
  Moreover, for $x \in \NN$, let $\mathcal{T}_x$ denote the set system $\{ D(\mathbf{T}[i], \mathbf{T}[n]) \mid i \in [n - 1], d(\mathbf{T}[i], \mathbf{T}[n]) = x \}$.
\end{definition}
The set system~$\mathcal T$ contains the subsets (without duplicates) of column indices corresponding to the columns where the row vectors~$\mathbf{T}[1],\ldots,\mathbf{T}[n-1]$ differ from~$\mathbf{T}[n]$.
For given~$\mathbf{T}[n]$, all the rows of~$\mathbf{T}$ can be determined from $\mathcal{T}$, as we have binary alphabet.

The concept of $\Delta$-systems has previously been used to obtain efficient algorithms \cite{EGKOS21,FG06,FBNS16}.
They are defined as follows (see also \Cref{fig:deltasystems}):

\begin{definition}[Weak $\Delta$-system]
  A set family $\mathcal{F} = \{ S_1, \dots, S_m \}$ is a \emph{weak $\Delta$-system} if there exists an integer $\lambda \in \NN$ such that $|S_i \cap S_j| = \lambda$ for any pair of distinct sets $S_i, S_j \in \mathcal{F}$.
  The integer~$\lambda$ is called the \emph{intersection size} of $\mathcal{F}$.
\end{definition}

\begin{definition}[Strong $\Delta$-system, Sunflower]
  A set family $\mathcal{F} = \{ S_1, \dots, S_m \}$ is a \emph{strong $\Delta$-system} (or \emph{sunflower}) if there exists a subset $C \subseteq S_1 \cup \dots \cup S_m$ such that $S_i \cap S_j = C$ for any pair of distinct sets $S_i, S_j \in \mathcal{F}$.
  We call the set $C$ the \emph{core} and the sets $P_i = S_i \setminus C$ the \emph{petals} of $\mathcal{F}$.
\end{definition}
Clearly, every strong $\Delta$-system is a weak $\Delta$-system.

\begin{figure}
  \centering
  \begin{tikzpicture}[scale=0.45]
    \node at (0.5, 3.7) {1};
    \node at (1.5, 3.7) {2};
    \node at (2.5, 3.7) {3};
    \fill[black] (0,0) rectangle (2,1);
    \fill[black] (0,1) rectangle (1,2);
    \fill[black] (2,1) rectangle (3,2);
    \fill[black] (1,2) rectangle (3,3);
    \draw[help lines] (0,0) grid (3,3);

    \node at (4.5, 0.5) {$\{ 1, 2 \}$};
    \node at (4.5, 1.5) {$\{ 1, 3 \}$};
    \node at (4.5, 2.5) {$\{ 2, 3 \}$};

    \node at (0.5, 3.7) {1};
    \node at (1.5, 3.7) {2};
    \node at (2.5, 3.7) {3};

    \begin{scope}[shift={(10, -.5)}]
      \fill[black] (0,0) rectangle (1,4);
      \fill[black] (1,3) rectangle (2,4);
      \fill[black] (2,2) rectangle (3,3);
      \fill[black] (3,1) rectangle (4,2);
      \fill[black] (4,0) rectangle (5,1);
      \draw[help lines] (0,0) grid (5,4);
      \node at (6.5, 0.5) {$\{ 1, 5 \}$};
      \node at (6.5, 1.5) {$\{ 1, 4 \}$};
      \node at (6.5, 2.5) {$\{ 1, 3 \}$};
      \node at (6.5, 3.5) {$\{ 1, 2 \}$};
      \node at (0.5, 4.7) {1};
      \node at (1.5, 4.7) {2};
      \node at (2.5, 4.7) {3};
      \node at (3.5, 4.7) {4};
      \node at (4.5, 4.7) {5};
    \end{scope}
  \end{tikzpicture}
  \caption{Illustration of a weak $\Delta$-system with intersection size one (left) and a strong $\Delta$-system with core $\{ 1 \}$ (right).}
  \label{fig:deltasystems}
\end{figure}

Our algorithms employ the combinatorial property that under certain conditions the set system~$\mathcal T$ of a matrix~$\mathbf T$ with bounded diameter forms a strong $\Delta$-system (which can be algorithmically exploited).
We say that a family $\mathcal{F}$ of sets is $h$-uniform if $|S| = h$ holds for each $S \in \mathcal{F}$.
Deza~\cite{Deza1974solution} showed that an $h$-uniform weak $\Delta$-system is a strong $\Delta$-system if its cardinality is sufficiently large (more precisely, if $|\mathcal{F}| \ge h^2 - h + 2$).
Moreover, Deza~\cite{deza1973propriete} also proved a stronger lower bound for uniform weak $\Delta$-systems in which the intersection size is exactly half of the cardinality of each set.
(We remark that our algorithm could rely on the weaker bound of Deza but using the stronger bound yields a faster algorithm.)

\begin{lemma}[{\cite[Th\'eor\`eme~1.1]{deza1973propriete}}]
  \label{prop:sunflower}
  Let $\mathcal{F}$ be a $(2 \mu)$-uniform weak $\Delta$-system with intersection size $\mu$.
  If $|\mathcal{F}| \ge \mu^2 + \mu + 2$, then $\mathcal{F}$ is a strong $\Delta$-system.
\end{lemma}

We extend this result to the case in which the set size is odd.

\begin{lemma}[restate=sunflowerodd,name=]
  \label{prop:sunflowerodd}
  Let $\mathcal{F}$ be a $(2 \mu + 1)$-uniform weak $\Delta$-system.
  \begin{enumerate}[label=(\roman*)]
  \item\label{mu+1} If the intersection size of~$\mathcal{F}$ is $\mu + 1$ and $|\mathcal{F}| \ge \mu^2 + \mu + 3$, then $\mathcal{F}$ is a strong $\Delta$-system.
    \item\label{mu} If the intersection size of~$\mathcal{F}$ is $\mu$ and $|\mathcal{F}| \ge (\mu+1)^2 + \mu + 3$, then $\mathcal{F}$ is a strong~$\Delta$-system.
  \end{enumerate}
\end{lemma}
\begin{proof}
  \ref{mu+1} Let $S \in \mathcal{F}$ and let $\mathcal{F}' = \{ T \triangle S \mid  T \in \mathcal{F} \setminus \{ S \} \}$.
  Here $T \triangle S$ denotes the symmetric difference $(T \setminus S) \cup (S \setminus T)$.
  Note that $\mathcal{F}'$ is a $(2 \mu)$-uniform weak $\Delta$-system with intersection size $\mu$:
  \begin{itemize}
    \item
      For each $T \in \mathcal{F} \setminus \{ S \}$, we have $|T \triangle S| = |S \setminus T| + |T \setminus S| = 2 \mu$.
    \item
      We show that $|(T \triangle S) \cap (U \triangle S)| = \mu$ for each distinct $T, U \in \mathcal{F} \setminus \{ S \}$.
      We rewrite
      \begin{align*}
        &|(T \triangle S) \cap (U \triangle S)| \\
        &= |((T \setminus S) \cup (S \setminus T)) \cap ((U \setminus S) \cup (S \setminus U))| \\
        &= |((T \setminus S) \cap (U \setminus S)) \cup ((T \setminus S) \cap (S \setminus U)) \\
        &\qquad\cup  ((S \setminus T) \cap (U \setminus S)) \cup ((S \setminus T) \cap (S \setminus U))| \\
        &= |((T \cap U) \setminus S) \cup (S \setminus (T \cup U)) | \\
        &= |((T \cap U) \setminus S)| + |(S \setminus (T \cup U))|.
      \end{align*}
      Here the third equality follows from ${(T \setminus S) \cap (S \setminus U) = (S \setminus T) \cap (U \setminus S) = \emptyset}$.
      Let $\kappa = |S \cap T \cap U|$.
      Since $|S \cap T| = |S \cap U| = \mu + 1$, it follows that $|(S \cap T) \setminus U| = |(S \cap U) \setminus T| = \mu - \kappa + 1$.
      Thus, we obtain
      \begin{align*}
        |S \setminus (T \cup U))|
        &= |S| - |(S \cap T) \setminus U| - |(S \cap U) \setminus T| - |S \cap T \cap U| \\
        &= (2 \mu + 1) - (\mu - \kappa + 1) - (\mu - \kappa + 1) - \kappa = \kappa - 1.
      \end{align*}
      Moreover, we obtain 
      \begin{align*}
        |((T \cap U) \setminus S)| = |T \cap U| - |S \cap T \cap U| = \mu - \kappa + 1.
      \end{align*}
      Now we have $|(T \triangle S) \cap (U \triangle S)| = |(S \setminus (T \cup U))| + |((T \cap U) \setminus S)| = \mu$.
  \end{itemize}
  Now, \Cref{prop:sunflower} implies that $\mathcal{F}'$ is a strong $\Delta$-system.
  Let $C'$ be the core of $\mathcal{F}'$.
  Note that $|(T \triangle S) \cap S| = |S \setminus T| = \mu$ for each $T \in \mathcal{F} \setminus \{ S \}$ or equivalently $|T' \cap S| = \mu$ for each $T' \in \mathcal{F}'$.
  We claim that $T' \cap S = C'$ for each $T' \in \mathcal{F}'$.
  Suppose not.
  Then, we have two cases: $C' \setminus S \ne \emptyset$ or $(T' \cap S) \setminus C' \ne \emptyset$.
  We show that $C' \setminus S \ne \emptyset$ holds for the latter case as well.
  Since $|S \cap C'| = |T' \cap S| - |(T' \cap S') \setminus C|$, we have $|S \cap C'| < |T' \cap S| = \mu$.
  This gives us $|C' \setminus S| = |C'| - |S \cap C'| > 0$.
  We thus have $C' \setminus S \ne \emptyset$.
  It follows that there exists an element $x \in (T' \setminus C') \cap S$ for each $T' \in \mathcal{F}'$.
  Since the set family $\{ T' \setminus C' \mid T' \in \mathcal{F}' \}$ is pairwise disjoint, it gives us $|S| \ge \mu^2 + \mu + 2 > 2 \mu + 1$, a contradiction.
  Thus, $\mathcal{F}$ is a sunflower with its core being $S \setminus C'$.

  \ref{mu} Let $x$ be an element which is not included in any set of $\mathcal{F}$.
  Consider the set family $\mathcal{F}' = \{ S \cup \{ x \} \mid S \in \mathcal{F} \}$.
  It is easy to see that $\mathcal{F}'$ is a $(2\mu+2)$-uniform weak $\Delta$-system with intersection size $\mu+1$.
  Since $\mathcal{F'}$ is a sunflower by \Cref{prop:sunflower}, so is $\mathcal{F}$.
\end{proof}

In order to obtain a linear-time algorithm for $(0, 2)$-\mc, we
will prove that for $\mathbf{T} \in \{ 0, 1 \}^{n \times \ell}$ with~$\maxd(\mathbf{T})\le 2$ and sufficiently large $\ell$, 
the set system $\mathcal{T}$ is a sunflower.
This yields a linear-time algorithm via a reduction to a linear-time solvable special case of \textsc{ConRMC}.
We start with a simple observation on matrices of diameter two, which will be helpful in the subsequent proofs.

\begin{obs}
  \label{obs:d2pattern}
  Let $\mathbf{T} \in \{ 0, 1 \}^{n \times \ell}$ be a matrix with $\maxd(\mathbf{T}) \le 2$.
  For each $T_1 \in \mathcal{T}_1$ and $T_2, T_2' \in \mathcal{T}_2$, it holds that $T_1 \subseteq T_2$ and that $|T_2 \cap T_2'| \ge 1$ (otherwise there exists a pair of rows with Hamming distance at least three).
\end{obs}

The next lemma states that~$|\mathcal{T}_2|$ restricts the number of columns.

\begin{lemma}
  \label{lemma:d2distinctcols}
  Let $\mathbf{T} \in \{ 0, 1 \}^{n \times \ell}$ be a matrix consisting of only dirty columns with $\maxd(\mathbf{T}) \le 2$.
  If $\mathcal{T}_2 \ne \emptyset$, then $\ell \le |\mathcal{T}_2| + 1$.
\end{lemma}
\begin{proof}
  First, observe that $\ell = |\bigcup_{T_1 \in \mathcal{T}_1} T_1 \cup \bigcup_{T_2 \in \mathcal{T}_2} T_2|$ because each column of $\mathbf{T}$ is dirty.
  Thus, it follows from \Cref{obs:d2pattern} that $\ell = |\bigcup_{T_2 \in \mathcal{T}_2} T_2|$.
  We prove the lemma by induction on~$|\mathcal{T}_2|$.
  Clearly, we have at most two columns if $|\mathcal{T}_2| = 1$.
  Suppose that $|\mathcal{T}_2| \ge 2$.
  For $T_2 \in \mathcal{T}_2$, we claim that
  \begin{align*}
    \ell
    = \Bigg|\bigcup_{T_2' \in \mathcal{T}_2} T_2' \, \Bigg|
    = \Bigg|\bigcup_{T_2' \in \mathcal{T}_2 \setminus \{ T_2 \}} T_2' \, \Bigg| + \Bigg| \, T_2 \mathbin{\Big\backslash} \bigcup_{T_2' \in \mathcal{T}_2 \setminus \{ T_2 \}} T_2' \, \Bigg|
    \le |\mathcal{T}_2| + 1.
  \end{align*}
  The induction hypothesis gives us that $|\bigcup_{T_2' \in \mathcal{T}_2 \setminus \{ T_2 \}} T_2'| \le |\mathcal{T}_2|$.
  For the other term, observe that $|T_2 \setminus \bigcup_{T_2' \in \mathcal{T}_2 \setminus \{ T_2 \}} T_2'| \le |T_2 \setminus T_2''| = |T_2| - |T_2 \cap T_2''|$ for $T_2'' \in \mathcal{T}_2 \setminus \{ T_2 \}$.
  Hence, it follows from \Cref{obs:d2pattern} that the second term is at most 1.
\end{proof}

Next, we show that a matrix with diameter at most two has radius at most one as long as it has at least five columns. Thus, we can solve \mc{} by solving \textsc{ConRMC} with radius one, which can be done in linear time via a reduction to \textsc{2-SAT}~\cite{KFN20}.
We use the following lemma concerning certain intersections of a set with elements of a sunflower.

\begin{lemma}[{\cite[Lemma~8]{FBNS16}}]
  \label{lemma:containcore}
  Let $\lambda \in \NN$, let $\mathcal{F}$ be a sunflower with core $C$, and let $X$ be a set such that $|X \cap S| \ge \lambda$ for all $S \in \mathcal{F}$.
  If $|\mathcal F| > |X|$, then $\lambda \le |C|$ and $|X \cap C| \ge \lambda$.
\end{lemma}

\begin{lemma}
  \label{lemma:d2structure}
  Let $\mathbf{T} \in \{ 0, 1 \}^{n \times \ell}$ be a matrix with $\maxd(\mathbf{T}) \le 2$.
  If $\ell \ge 5$, then there exists a vector $v \in \{ 0, 1 \}^\ell$ such that $\dist(v, \mathbf{T}[i]) \le 1$ for all $i \in [n]$.
\end{lemma}
\begin{proof}
  If $\mathcal{T}_2 = \emptyset$, then we are immediately done by definition, because $\dist(\mathbf{T}[n], \mathbf{T}[i]) \le 1$ for all $i \in [n]$ (see \Cref{fig:d2structure1} for an illustration).
  Since $\ell \ge 5$, \Cref{lemma:d2distinctcols} implies $|\mathcal{T}_2| \ge 4$.

  It follows from \Cref{obs:d2pattern} that $\mathcal{T}_2$ is a 2-uniform weak $\Delta$-system with intersection size~one (see \Cref{fig:d2structure2}).
  Thus,~$\mathcal{T}_2$ is a sunflower by \Cref{prop:sunflower}.
  Let $\{ j_{\core} \}$ denote the core of~$\mathcal{T}_2$.
  Note that $|T_1 \cap T_2| \ge 1$ holds for each $T_1 \in \mathcal{T}_1$ and $T_2 \in \mathcal{T}_2$ by \Cref{obs:d2pattern}.
  Now we can infer from \Cref{lemma:containcore} (let $X = T_1$, $\lambda = 1$, and $\mathcal{F} = \mathcal{T}_2$) that $\mathcal{T} \subseteq \{ T_1 \}$, where $T_1 = \{ j_{\core} \}$.

  Hence, it holds that $\dist(v, \mathbf{T}[i]) \le 1$ for all $i \in [n]$, where $v \in \{ 0, 1 \}^\ell$ is a row vector such that $v[j_{\core}] = 1 - \mathbf{T}[n, j_{\core}]$ and $v[j] = \mathbf{T}[n, j]$ for each $j \in [\ell] \setminus \{ j_{\core} \}$.
\end{proof}


\begin{figure}[t]
  \centering
  \begin{subfigure}{.32\textwidth}
    \centering
    \begin{tikzpicture}[scale=0.45]
      \fill[black] (4,5) rectangle (5,6);
      \fill[black] (3,4) rectangle (4,5);
      \fill[black] (0,3) rectangle (1,4);
      \fill[black] (2,2) rectangle (3,3);
      \fill[black] (1,1) rectangle (2,2);
      \draw[help lines] (0,0) grid (5,6);
    \end{tikzpicture}
    \caption{The case $\mathcal{T}_2 = \emptyset$.}
    \label{fig:d2structure1}
  \end{subfigure}
  \begin{subfigure}{.32\textwidth}
    \centering
    \begin{tikzpicture}[scale=0.45]
      \fill[black] (4,5) rectangle (5,6);
      \fill[black] (3,4) rectangle (4,5);
      \fill[black] (0,3) rectangle (1,4);
      \fill[black] (2,2) rectangle (3,3);
      \fill[black] (1,1) rectangle (2,6);
      \draw[help lines] (0,0) grid (5,6);
    \end{tikzpicture}
    \caption{The case $|\mathcal{T}_2| \ge 4$.}
    \label{fig:d2structure2}
  \end{subfigure}
  \caption{
    Illustration of \Cref{lemma:d2structure} with $n = 6$.
    A black cell denotes a value different from row~$\mathbf{T}[6]$.
    In~(b) the set system~$\mathcal{T}_2$ forms a sunflower with core~$\{2\}$.
    In both cases the radius is one.
  }
  \label{fig:d2structure}
\end{figure}

\begin{theorem}
  \label{theorem:d2linear}
  $(0, 2)$-\mc{} can be solved in $O(n \ell)$ time.
\end{theorem}
\begin{proof}
  Let $\mathbf{S} \in \{ 0, 1, \square \}^{n \times \ell}$ be an input matrix of $(0, 2)$-\mc. 
  If $\ell \le 4$, then we use the linear-time algorithm of \Cref{lemma:linearforconstantl}.
  Henceforth, we assume that $\ell \ge 5$.

  We claim that $\mathbf{S}$ is a \textbf{Yes}-instance if and only if the \textsc{ConRMC} instance $I = (\mathbf{S}, 1^n)$ is a \textbf{Yes}-instance.

  $(\Rightarrow)$
  Let $\mathbf{T}$ be a completion of $\mathbf{S}$ with $\maxd(\mathbf{T}) \le 2$.
  Since $\ell \ge 5$, there exists a vector $v$ such that $\dist(v, \mathbf{T}[i]) \le 1$ for all $i \in [n]$ by \Cref{lemma:d2structure}.
  It follows that $I$ is a \textbf{Yes}-instance.

  $(\Leftarrow)$
  Let $v$ be a solution of $I$.
  Let $\mathbf{T}$ be the matrix such that for each $i \in [n]$, $\mathbf{T}[i] = \mathbf{S}[i] \leftarrow v$ (recall that $u \leftarrow v$ denotes the vector obtained from $u$ by replacing all missing entries of $u$ with the entries of $v$ in the corresponding positions).
  Then, we have $\dist(v, \mathbf{T}[i]) \le 1$ for each $i \in [n]$.
  By the triangle inequality, we obtain $\dist(\mathbf{T}[i], \mathbf{T}[i']) \le \dist(v, \mathbf{T}[i]) + \dist(v, \mathbf{T}[i']) \le 2$ for each $i, i' \in [n]$.

  Since \textsc{ConRMC} can be solved in linear time when $\max_{i \in [n]} r[i] = 1$ \cite[Theorem~1]{KFN20}, it follows that $(0, 2)$-\mc{} can be solved in linear time.
\end{proof}

We next show polynomial-time solvability of $(0,3)$-\mc{} (\Cref{theorem:d3polynomial}).
The overall idea is, albeit technically more involved, similar to $(0, 2)$-\mc.
We first show that the set family~$\mathcal{T}$ of a matrix~$\mathbf{T}$ with $\delta(\mathbf{T}) = 3$ contains a sunflower by \Cref{prop:sunflowerodd}.
We then show that such a matrix has a certain structure which again allows us to reduce the problem to the linear-time solvable special case of~\textsc{ConRMC} with radius one.

We start with an observation on a matrix whose diameter is at most three.

\begin{obs}
  \label{obs:d3pattern}
  Let $\mathbf{T} \in \{ 0, 1 \}^{n \times \ell}$ be a matrix with $\maxd(\mathbf{T}) \le 3$.
  It holds for each $T_1 \in \mathcal{T}_1$, $T_2 \in \mathcal{T}_2$, and $T_3, T_3' \in \mathcal{T}_3$ that $T_1 \subseteq T_3$, $T_2 \cap T_3 \ne \emptyset$, and $|T_3 \cap T_3'| \ge 2$ (otherwise there exists a pair of rows with Hamming distance four).
\end{obs}

From \Cref{obs:d3pattern}, we obtain (by induction) the following lemma analogously to \Cref{lemma:d2distinctcols}.

\begin{lemma}
  \label{lemma:d3distinctcols}
  Let $\mathbf{T} \in \{ 0, 1 \}^{n \times \ell}$ be a matrix consisting of dirty columns with $\maxd(\mathbf{T}) \le 3$.
  If $\mathcal{T}_3 \neq \emptyset$, then $\ell \le |\mathcal{T}_2| + |\mathcal{T}_3| + 2$.
\end{lemma}
\begin{proof}
  First, observe that $\ell = |\bigcup_{T_1 \in \mathcal{T}_1} T_1 \cup \bigcup_{T_2 \in \mathcal{T}_2} T_2 \cup \bigcup_{T_3 \in \mathcal{T}_3} T_3|$ because each column of $\mathbf{T}$ is dirty.
  Thus, it follows from \Cref{obs:d2pattern} that $\ell = |\bigcup_{T_2 \in \mathcal{T}_2} T_2 \cup \bigcup_{T_3 \in \mathcal{T}_3} T_3|$.
  We prove the lemma by induction on~$|\mathcal{T}_2| + |\mathcal{T}_3|$.
  We have at most three columns if $|\mathcal{T}_2| + |\mathcal{T}_3| = 1$.
  Suppose that $|\mathcal{T}_2| + |\mathcal{T}_3| \ge 2$.
  For $T \in \mathcal{T}_2 \cup \mathcal{T}_3$ of minimize size, we claim that
  \begin{align*}
    \ell
    = \Bigg|\bigcup_{T' \in \mathcal{T}_2 \cup \mathcal{T}_3} T' \, \Bigg|
    = \Bigg|\bigcup_{T' \in \mathcal{T}_2 \cup \mathcal{T}_3 \setminus \{ T \}} T' \, \Bigg| + \Bigg| \, T \mathbin{\Big\backslash} \bigcup_{T' \in \mathcal{T}_2 \cup  \mathcal{T}_3 \setminus \{ T \}} T' \, \Bigg|
    \le |\mathcal{T}_2| + |\mathcal{T}_3| + 2.
  \end{align*}
  The induction hypothesis gives us that $|\bigcup_{T' \in \mathcal{T}_2 \cup \mathcal{T}_3 \setminus \{ T \}} T'| \le |\mathcal{T}_2| + |\mathcal{T}_3| + 1$.
  For the other term, since we can assume that $\mathcal{T}_2 \cup \mathcal{T}_3 \setminus \{ T \}$ has a set $T_3$ from $\mathcal{T}_3$, we have $|T \setminus \bigcup_{T \in \mathcal{T}_2 \cup \mathcal{T}_3 \setminus \{ T \}} T'| \le |T \setminus T_3| = |T| - |T \cap T_3|$.
  Hence, it follows from \Cref{obs:d2pattern} that the second term is at most 1.
\end{proof}

Our goal is to use \Cref{prop:sunflowerodd} to derive that $\mathcal{T}_3$ forms a sunflower, that is, we need that $|\mathcal{T}_3|\ge 5$.
The next lemma shows that this holds when~$\mathbf{T}$ has at least 14 dirty columns. 

\begin{lemma}
  \label{lemma:fived3pairs}
  Let $\mathbf{T} \in \{ 0, 1 \}^{n \times \ell}$ be a matrix consisting of $\ell \ge 14$ dirty columns with $\maxd(\mathbf{T}) = 3$.
  Then, $|\mathcal{T}_3| \ge 5$ (for some reordering of rows).
\end{lemma}
\begin{proof}
  Assume that the rows are reordered such that $|\mathcal{T}_3|$ is maximized.
  If $|\mathcal{T}_3| \le 4$, then we have $|\mathcal{T}_2| \ge 8$ by \Cref{lemma:d3distinctcols}.
  Let $T_3 \in \mathcal{T}_3$.
  By \Cref{obs:d3pattern}, $T_2 \cap T_3 \neq \emptyset$ holds for each $T_2 \in \mathcal{T}_2$.
  There are at most three sets $T_2\in\mathcal{T}_2$ with $T_2 \subseteq T_3$.
  Thus, there are at least five sets $T_2\in\mathcal{T}_2$ such that $|T_2 \cap T_3| = 1$.
  For each of these five sets, it holds that $|T_2 \triangle T_3| = |T_2| + |T_3| - 2|T_2 \cap T_3| = 3$.
  This contradicts the choice of the row reordering.
\end{proof}

With \Cref{lemma:fived3pairs} at hand, we are ready to reveal the structure of a diameter-three matrix (see \Cref{fig:d3structure} for an illustration).

\begin{lemma}
  \label{lemma:d3structure}
  Let $\mathbf{T} \in \{ 0, 1 \}^{n \times \ell}$ be a matrix with $\delta(\mathbf{T}) \le 3$ and $|\mathcal{T}_3| \ge 5$.
  Then, there exist $j_1 \ne j_2 \in [\ell]$ such that the following hold:
  \begin{itemize}
    \item
      $T_1 \subseteq \{ j_1, j_2 \}$ for each $T_1 \in \mathcal{T}_1$.
    \item
      $T_2 \cap \{ j_1, j_2 \} \ne \emptyset$ for each $T_2 \in \mathcal{T}_2$.
    \item
      $T_3 \supseteq \{ j_1, j_2 \}$ for each $T_3 \in \mathcal{T}_3$.
  \end{itemize}
  Moreover, exactly one of the following holds for $\mathcal{T}_2' = \{ T_2 \in \mathcal{T}_2 \mid j_1 \in T_2 \wedge j_2 \not\in T_2 \}$ and  $\mathcal{T}_2'' = \{ T_2 \in \mathcal{T}_2 \mid j_1 \not\in T_2 \wedge j_2 \in T_2 \}$:
  \begin{enumerate}[label={(\alph*)}]
    \item
      \label{lemma:d3structure1}
      $\mathcal{T}_2' = \emptyset$ or $\mathcal{T}_2'' = \emptyset$.
    \item
      \label{lemma:d3structure2}
      $\mathcal{T}_2' = \{ \{ j_1, j_3 \} \}$ and $\mathcal{T}_2'' = \{ \{ j_2, j_3 \} \}$ for some $j_3 \in [\ell]$.
    \end{enumerate}
\end{lemma}
\begin{proof}
  Note that $\mathcal{T}_3$ is $3$-uniform by definition and note also that it is a weak $\Delta$-system with intersection size~two by \Cref{obs:d3pattern}.
  Hence, $\mathcal{T}_3$ is a strong $\Delta$-system by \Cref{prop:sunflowerodd}.
  Let $\{ j_1, j_2 \}$ be its core.
  Then, we have $\{ j_1, j_2 \} \subseteq T_3$ for each $T_3 \in \mathcal{T}_3$.
  It follows from \Cref{obs:d3pattern} and \Cref{lemma:containcore} that $T_1 \subseteq \{ j_1, j_2 \}$ for each $T_1 \in \mathcal{T}_1$ and $T_2 \cap \{ j_1, j_2 \} \ne \emptyset$ for each $T_2 \in \mathcal{T}_2$.

  Now we show that either \ref{lemma:d3structure1} or \ref{lemma:d3structure2} holds.
  Suppose that $|\mathcal{T}_2'| \ge 2$ and $|\mathcal{T}_2''| \ge 1$, and let $T_2 \ne T_2' \in \mathcal{T}_2'$ and $T_2'' \in \mathcal{T}_2''$.
  Then, either $T_2 \cap T_2'' = \emptyset$ or $T_2' \cap T_2'' = \emptyset$ must hold.
  However, this is a contradiction because the corresponding row vectors have Hamming distance~four.
  Thus, we have that $|\mathcal{T}_2'| \le 1$ or $\mathcal{T}_2'' = \emptyset$.
  Analogously, we obtain $\mathcal{T}_2' = \emptyset$ or $|\mathcal{T}_2''| \le 1$.
  If $\mathcal{T}_2' = \emptyset$ or $\mathcal{T}_2'' = \emptyset$, then \ref{lemma:d3structure1} is satisfied.
  Otherwise, we have $|\mathcal{T}_2'| = |\mathcal{T}_2''| = 1$.
  Since $|T_2' \triangle T_2''| = |T_2'| + |T_2''| - 2|T_2' \cap T_2''| \le \delta(\mathbf{T}) \le 3$, we obtain $T_2' \cap T_2'' \ne \emptyset$ for each $T_2' \in \mathcal{T}_2'$ and $T_2'' \in \mathcal{T}_2''$.
  Hence, \ref{lemma:d3structure2} holds.
\end{proof}

\begin{figure}[t]
  \centering
  \begin{subfigure}{.4\textwidth}
    \centering
    \begin{tikzpicture}[scale=0.45]
      \pgfmathsetmacro{\west}{0}
      \pgfmathsetmacro{\east}{7}
      \pgfmathsetmacro{\north}{0}
      \pgfmathsetmacro{\south}{8}
      \pgfmathsetmacro{\padding}{0.2}
      
      \fill[black] (0,1) rectangle (1,7);
      \fill[black] (1,1) rectangle (2,5);
      \fill[black] (1,7) rectangle (2,8);
      \fill[black] (2,5) rectangle (3,6);
      \fill[black] (3,3) rectangle (4,4);
      \fill[black] (4,4) rectangle (5,5);
      \fill[black] (5,6) rectangle (6,7);
      \fill[black] (5,2) rectangle (6,3);
      \fill[black] (6,1) rectangle (7,2);
      \draw[help lines] (\west, \south) grid (\east ,\north);
      \draw[decoration={brace,raise=3pt},decorate] (\west,7) -- node[left=3pt] {$\mathcal{T}_1$} (\west,8);
      \draw[decoration={brace,raise=3pt},decorate] (\west,5) -- node[left=3pt] {$\mathcal{T}_2$} (\west,7);
      \draw[decoration={brace,raise=3pt},decorate] (\west,1) -- node[left=3pt] {$\mathcal{T}_3$} (\west,5);
      \fill[black] (0.5, \south + .5) node {$j_1$};
      \fill[black] (1.5, \south + .5) node {$j_2$};
    \end{tikzpicture}
  \end{subfigure}
  \quad
  \begin{subfigure}{.4\textwidth}
    \centering
    \begin{tikzpicture}[scale=0.45]
      \pgfmathsetmacro{\west}{0}
      \pgfmathsetmacro{\east}{6}
      \pgfmathsetmacro{\north}{8}
      \pgfmathsetmacro{\south}{0}
      \pgfmathsetmacro{\padding}{0.2}

      \fill[black] (0.5,8.5) node {$j_1$};
      \fill[black] (1.5,8.5) node {$j_2$};
      \fill[black] (2.5,8.5) node {$j_3$};

      \fill[black] (1,7) rectangle (2,8);

      \fill[black] (0,6) rectangle (1,7);
      \fill[black] (2,6) rectangle (3,7);

      \fill[black] (1,5) rectangle (2,6);
      \fill[black] (2,5) rectangle (3,6);

      \fill[black] (0,1) rectangle (2,5);

      \fill[black] (3,4) rectangle (4,5);
      \fill[black] (2,3) rectangle (3,4);
      \fill[black] (4,2) rectangle (5,3);
      \fill[black] (5,1) rectangle (6,2);
      \draw[help lines] (\west, \south) grid (\east ,\north);
      \draw[decoration={brace,raise=3pt},decorate] (\west,7) -- node[left=3pt] {$\mathcal{T}_1$} (\west,8);
      \draw[decoration={brace,raise=3pt},decorate] (\west,5) -- node[left=3pt] {$\mathcal{T}_2$} (\west,7);
      \draw[decoration={brace,raise=3pt},decorate] (\west,1) -- node[left=3pt] {$\mathcal{T}_3$} (\west,5);
    \end{tikzpicture}
  \end{subfigure}
  \caption{
    Illustration (for smaller $\ell$) of \Cref{lemma:d3structure}~\ref{lemma:d3structure1} (left) and \ref{lemma:d3structure2} (right).
   A black cell indicates that the entry differs from the last row vector in the corresponding column.
  }
  \label{fig:d3structure}
\end{figure}

The following lemma establishes a connection to \textsc{ConRMC}.
For $v, v' \in \{ 0, 1\}^\ell$ and $J \subseteq [\ell]$, we write $\dist_J(v, v')$ to denote $\dist(v[J], v'[J])$.

\begin{lemma}
  \label{lemma:d3v}
  Let $\mathbf{T} \in \{ 0, 1 \}^{n \times \ell}$ be a matrix consisting of dirty columns with $\maxd(\mathbf{T}) = 3$.
  If $\ell \ge 14$, then there exists $v \in \{ 0, 1 \}^\ell$ such that at least one of the following holds:
  \begin{enumerate}[label={(\alph*)}]
    \item \label{lemma:d3v1}
      There exists $j \in [\ell]$ such that $\dist_{[\ell] \setminus \{ j \}}(v, \mathbf{T}[i]) \le 1$ for all $i \in [n]$.
    \item \label{lemma:d3v2}
      There exist three column indices $J = \{ j_1, j_2, j_3 \} \subseteq [\ell]$ such that all of the following hold for each $i \in [n]$:
      \begin{itemize}
        \item
          $\dist(v_J, t_{i, J}) \le 2$.
        \item
          If $\dist(v_J, t_{i, J}) \ge 1$, then $\dist_{[\ell] \setminus J}(v, \mathbf{T}[i]) = 0$.
        \item
          If $\dist(v_J, t_{i, J}) = 0$, then $\dist_{[\ell] \setminus J}(v, \mathbf{T}[i]) \le 1$.
      \end{itemize}
      Here $v_J = (v[j_1], v[j_2], v[j_3])$ and $t_{i, J} = (\mathbf{T}[i, j_1], \mathbf{T}[i, j_2], \mathbf{T}[i, j_3])$ for each $i \in [n]$.
  \end{enumerate}
\end{lemma}
\begin{proof}
  From \Cref{lemma:fived3pairs}, we can assume that $|\mathcal{T}_3| \ge 5$.
  Hence, \Cref{lemma:d3structure} applies. Let~$j_1$ and $j_2$ be the according column indices.
  Let $v \in \{ 0, 1 \}^\ell$ be the row vector with
  \begin{align*}
    v[j] = \begin{cases}
      1 - \mathbf{T}[n, j] & \text{if } j \in \{ j_1, j_2 \}, \\
      \mathbf{T}[n, j] & \text{otherwise}, \\
    \end{cases}
  \end{align*}
  for each $j \in [\ell]$.
  We claim that \ref{lemma:d3v1} corresponds to \Cref{lemma:d3structure}~\ref{lemma:d3structure1}, and \ref{lemma:d3v2}
  corresponds to \Cref{lemma:d3structure}~\ref{lemma:d3structure2}.

  Suppose that \Cref{lemma:d3structure}~\ref{lemma:d3structure1} holds with $\mathcal{T}_2'' = \emptyset$ (the case $\mathcal{T}_2' = \emptyset$ is completely analogous).
  We prove that $\dist_{[\ell] \setminus \{ j_2 \}}(v, \mathbf{T}[i]) \le 1$ for all $i \in [n]$.
  Since $\dist_{[\ell] \setminus \{ j_2 \}}(v, \mathbf{T}[i]) = |\{ j_1 \} \triangle (D(\mathbf{T}[i], \mathbf{T}[n]) \setminus \{ j_2 \})|$, it suffices to show that $|\{ j_1 \} \triangle (T \setminus \{ j_2 \})| \le 1$ holds for all $T \in \mathcal{T}$.
  Due to \Cref{lemma:d3structure}, we have
  \begin{itemize}
    \item
      $\{ j_1 \} \triangle (T \setminus \{ j_2 \}) \subseteq \{ j_1 \}$ for each $T \in \mathcal{T}_1$ and
    \item 
       $|\{ j_1 \} \triangle (T \setminus \{ j_2 \})| = |T \setminus \{ j_1, j_2 \}| \le 1$ for each $T \in \mathcal{T}_2 \cup \mathcal{T}_3$ since $j_1 \in T$.
  \end{itemize}
  Hence, \ref{lemma:d3v1} is true.
  
  Now, assume that \Cref{lemma:d3structure}~\ref{lemma:d3structure2} holds and let~$J=\{j_1,j_2,j_3\}$.
    If there exists $i \in [n]$ with $\dist(v_J, t_{i, J}) = 3$, then this implies $\{ j_3 \} \in \mathcal{T}_1$ which yields the contradiction $\{j_3\}\not\subseteq\{j_1,j_2\}$.
    
      Further, for each $T \in \mathcal{T}_1 \cup \mathcal{T}_2$, we have $T \setminus J = \emptyset$.
      Hence, for the row vector $\mathbf{T}[i]$ corresponding to~$T$, we have $\dist_{[\ell] \setminus J}(v, \mathbf{T}[i]) = 0$.
      Now, let $T_3 \in \mathcal{T}_3$ with corresponding row vector $\mathbf{T}[i]$.
      If $T_3 = \{ j_1, j_2, j_3 \}$, then $\dist(v_J, t_{i, J}) = 1$ and $\dist_{[\ell] \setminus J}(v, \mathbf{T}[i]) = 0$.
      Otherwise, we have $T_3 = \{ j_1, j_2, j \}$ for some $j \in [\ell] \setminus J$. Hence, $\dist(v_J, t_{i, J}) = 0$ and $\dist_{[\ell] \setminus J}(v, \mathbf{T}[i]) = |T_3 \setminus \{ j_1, j_2 \}| = 1$. Hence, \ref{lemma:d3v2} is true.
\end{proof}

Based on the connection to \textsc{ConRMC}, we obtain a polynomial-time algorithm.


\begin{restatable}{theorem}{dpolynomial}
  \label{theorem:d3polynomial}
  $(0, 3)$-\mc{} can be solved in $O(n \ell^4)$ time.
\end{restatable}

\begin{proof}
  We first apply \Cref{theorem:d2linear} to determine whether there exists a completion $\mathbf{T}\in\{0,1\}^{n\times\ell}$ of $\mathbf{S}\in\{0,1,\square\}^{n\times\ell}$ such that $\maxd(\mathbf{T}) \le 2$.
  If not, then it remains to determine whether there exists a completion~$\mathbf{T}$ with $\maxd(\mathbf{T}) = 3$.
  We can assume that $\ell \ge 14$ by \Cref{lemma:linearforconstantl}.
  We solve the problem by solving several instances of \textsc{ConRMC} based on \Cref{lemma:d3v}.

  For $j \in [\ell]$, let $I_{j}=(\mathbf{S}[:, [\ell] \setminus \{ j \}], 1^n)$ be a \textsc{ConRMC} instance and let $\mathcal{I}_1 = \{ I_j \mid j \in [\ell] \}$.
  These instances correspond to \Cref{lemma:d3v}~\ref{lemma:d3structure1}.

  Now, we describe the instances corresponding to \Cref{lemma:d3v}~\ref{lemma:d3structure2}.
  Let $j_1, j_2, j_3 \in [\ell]$ be three distinct column indices and let $v_1, v_2, v_3 \in \{ 0, 1\}$.
  We define an instance $I_{j_1, j_2, j_3}^{v_1, v_2, v_3} = (\mathbf{S}_{j_1, j_2, j_3}, r)$ of \textsc{ConRMC}
   as follows:
  \begin{itemize}
    \item
      $\mathbf{S}_{j_1, j_2, j_3} = \mathbf{S}[:, [\ell] \setminus \{ j_1, j_2, j_3 \}]$.
    \item
      For each $i \in [n]$, let
      \begin{align*}
        r[i] = \begin{cases}
          0 & \text{if $(\mathbf{S}[i, j_1] = 1-v_1) \vee (\mathbf{S}[i, j_2] = 1-v_2) \vee (\mathbf{S}[i, j_3] = 1-v_3)$.} \\
          1 & \text{otherwise}.
        \end{cases}
      \end{align*}
  \end{itemize}
  Let $\mathcal{I}_2$ contain those instances $I_{j_1, j_2, j_3}^{v_1, v_2, v_3}$ in which for each $i \in [n]$ at least one of $\mathbf{S}[i, j_1] \ne 1 - v_1$, $\mathbf{S}[i, j_2] \ne 1 - v_2$, or $\mathbf{S}[i, j_3] \ne 1 - v_3$ holds.
  We claim that $\mathbf{S}$ is a \textbf{Yes}-instance if and only if at least one instance in $\mathcal{I}_1$ or $\mathcal{I}_2$ is a \textbf{Yes}-instance.

  \begin{itemize}
    \item
      If $I_j \in \mathcal{I}_1$ is a \textbf{Yes}-instance, then there exists $v \in \{ 0, 1 \}^{\ell - 1}$ such that $\dist(v', \mathbf{S}[i, [\ell] \setminus \{ j \}]) \le 1$ for each $i \in [n]$.
      Let $\mathbf{T}$ be the completion of $\mathbf{S}$ in which $\mathbf{T}[i, [\ell] \setminus \{ j \}] = \mathbf{S}[i, [\ell] \setminus \{ j \}] \leftarrow v$ for each $i \in [n]$.
      Then, we have
      \begin{align*}
        \dist(\mathbf{T}[i], \mathbf{T}[i'])
        &\le \dist(\mathbf{T}[i, [\ell] \setminus \{ j \}], \mathbf{T}[i', [\ell \setminus \{ j \}]]) + 1 \\
        &\le \dist(v, \mathbf{T}[i, [\ell] \setminus \{ j \}]) + \dist(v, \mathbf{T}[i', [\ell \setminus \{ j \}]]) + 1 \le 3
      \end{align*}
      for each $i, i' \in [n]$.
    \item
      If $I_{j_1, j_2, j_3}^{v_1, v_2, v_3}=(\mathbf{S}_{j_1, j_2, j_3}, r)\in\mathcal{I}_2$ is a \textbf{Yes}-instance with solution $v' \in \{ 0, 1 \}^{\ell - 3}$, then
      let $v \in \{ 0, 1 \}^{\ell}$ be the row vector obtained from $v'$ by inserting $v_1$, $v_2$, and $v_3$ in the $j_1$-th, $j_2$-th, and $j_3$-th column, respectively, and 
      let $\mathbf{T}$ be the completion of $\mathbf{S}$ in which $\mathbf{T}[i] = \mathbf{S}[i] \leftarrow v$ or each $i \in [n]$.
      We prove that $\maxd(\mathbf{T}) \le 3$.
      Let $R_x \subseteq [n]$ be the set of row indices $i$ with $r[i] = x$ for $x \in \{ 0, 1 \}$.
      Then, we have that
      \begin{itemize}
        \item
          $\dist(v, \mathbf{T}[i]) = \dist_{\{ j_1, j_2, j_3 \}}(v, \mathbf{T}[i]) + \dist_{[\ell] \setminus \{ j_1, j_2, j_3 \}}(v, \mathbf{T}[i]) \le 2 + 0 = 2$ for each $i \in R_0$ and
        \item
          $\dist(v, \mathbf{T}[i]) = \dist_{\{ j_1, j_2, j_3 \}}(v, \mathbf{T}[i]) + \dist_{[\ell] \setminus \{ j_1, j_2, j_3 \}}(v, \mathbf{T}[i]) \le 0 + 1 = 1$ for each $i \in R_1$.
      \end{itemize}
      By the triangle inequality, we obtain $\dist(\mathbf{T}[i], \mathbf{T}[i']) \le 3$ for each $i, i' \in [n]$ with $i \in R_1$ or $i' \in R_1$.
      Thus, it suffices to show $\maxd(\mathbf{T}[R_0]) \le 3$. 
      Since $\mathbf{T}[i, [\ell] \setminus \{ j_1, j_2, j_3 \}] = \mathbf{T}[i', [\ell] \setminus \{ j_1, j_2, j_3 \}] = v'$ for each $i, i' \in R_0$, this clearly holds.
  \end{itemize}
  The reverse direction is easily verified using \Cref{lemma:d3v}.
  
  Overall, we construct $O(\ell^3)$ \textsc{ConRMC} instances, each of which can be solved in $O(n \ell)$ time~\cite[Theorem~1]{KFN20}.
  Hence, $(0, 3)$-\mc{} can be solved in $O(n \ell^4)$ time.
\end{proof}

Our algorithms work via reductions to \textsc{ConRMC}.
Although \textsc{ConRMC} imposes an upper bound on the diameter implicitly by the triangle inequality, it is seemingly difficult to enforce any lower bounds (that is, $\alpha > 0$).
In the next subsection, we will see polynomial-time algorithms for $\alpha > 0$, based on reductions to the graph factorization problem.

\subsection{Polynomial time for~\texorpdfstring{$\beta = \alpha +1$}{beta = alpha + 1}}
\label{sec:const:poly}

We now give polynomial-time algorithms for $(\alpha, \beta)$-\mc{} with constant~$\alpha > 0$ given that $\beta \le \alpha + 1$.
As in \Cref{sec:alpha0}, our algorithms exploit combinatorial structures revealed by Deza's theorem (\Cref{prop:sunflower,prop:sunflowerodd}).
Recall that $\mathcal{T}$ denotes a set system obtained from a complete matrix $\mathbf{T}$ (\Cref{def:diffsets}).
We show that $\mathcal{T}$ essentially is a sunflower when $\mind(\mathbf{T}) \ge \alpha$ and $\maxd(\mathbf{T}) \le \alpha + 1$.
For the completion into such a sunflower, it suffices to solve the following matrix completion problem, which we call \textsc{Sunflower Matrix Completion}.

\dprob
{Sunflower Matrix Completion (SMC)}
{An incomplete matrix $\mathbf{S} \in \{ 0, 1, \square \}^{n \times \ell}$ and $s, m \in \NN$.}
{Is there a completion $\mathbf{T} \in \{ 0, 1 \}^{n \times \ell}$ of $\mathbf{S}$ such that $D(\mathbf{T}[1], \mathbf{T}[n]), \dots, D(\mathbf{T}[n - 1], \mathbf{T}[n])$ are pairwise disjoint sets each of size at most $s$ and $\sum_{i \in [n-1]} |D(\mathbf{T}[i], \mathbf{T}[n])| \ge m$.}
Intuitively speaking, the problem asks for a completion into a sunflower with empty core and bounded petal sizes.
All algorithms presented in this subsection are via reductions to \textsc{SMC}.
First, we show that \textsc{SMC} is indeed polynomial-time solvable.
We prove this using a well-known polynomial-time algorithm for the graph problem~\textsc{$(g, f)$-Factor}~\cite{Gab83}.

\dprob
{$(g, f)$-Factor}
{A graph $G = (V, E)$, functions $g, f \colon V \to \NN$, and $m' \in \NN$.}
{Does $G$ contain a subgraph $G' = (V, E')$ such that $|E'| \ge m'$ and $g(v) \le \deg_{G'}(v) \le f(v)$ for all $v \in V$?}

\begin{lemma}
  \label{lemma:pmc}
  For constant~$s>0$, {\normalfont\textsc{SMC}} can be solved in $O(n\ell \sqrt{n + \ell})$ time.
\end{lemma}
\begin{proof}
 Let~$(\mathbf{S},s,m)$ be an \textsc{SMC} instance.
  Let~$a_j^{x}$ be the number of occurrences of $x\in\{0,1\}$ in $\mathbf{S}[:,j]$ for each $j \in [\ell]$.
  We can assume that $a_j^0 \ge a_j^1$ for each $j \in [\ell]$ (otherwise swap the occurrences of 0's and 1's in the column). If $a_j^0 \ge 2$ and $\mathbf{S}[n, j] = 1$ for some~$j\in[\ell]$, then we can return \textbf{No} since there will be two intersecting sets.
  Also, if $a_j^1 \ge 2$, then we return \textbf{No}.
  
  We construct an instance of \textsc{$(g, f)$-Factor} as follows.
  We introduce a vertex $u_i$ for each $i \in [n - 1]$ and a vertex $v_j$ for each $j \in [\ell]$.
  The resulting graph $G$ will be a bipartite graph with one vertex subset $\{ u_1, \dots, u_{n - 1} \}$ representing rows and the other $\{ v_1, \dots, v_{\ell} \}$ representing columns.
  Essentially, we add an edge between $u_i$ and $v_j$ if the column $\mathbf{S}[:, j]$ can be completed such that the $i$-th entry differs from all other entries on $\mathbf{S}[:, j]$ (see \Cref{fig:factor} for an illustration). Intuitively, such an edge encodes the information that column index~$j$ can be contained in a petal of the sought sunflower.
  Formally, there is an edge $\{u_i,v_j\}$ if and only if there is a completion $t_j\in\{0,1\}^n$ of $\mathbf{S}[:, j]$ in which $t_j[h] = 1 - t_j[i]$ for all $h \in [n - 1] \setminus \{ i \}$.
  We set $g(u_i) := 0$ and $f(u_i) := s$ for each $i \in [n - 1]$, $g(v_j) := a_j^1$ and $f(v_j) := 1$ for each $j \in [\ell]$, and $m' := m$.
  This construction can be done in~$O(n\ell)$ time.
  To see this, note that the existence of an edge~$\{u_i,v_j\}$ only depends on~$a_j^0$,~$a_j^1$, and $\mathbf{S}[i, j]$.
  \begin{itemize}
    \item
      If $a_j^0 \le 1$ and $a_j^1 = 0$, then add the edge $\{u_i,v_j\}$.
      The corresponding completion~$t_j$ can be seen as follows:
      \begin{itemize}
        \item
          If $\mathbf{S}[h, j] = \square$ for all $h \in [n - 1]$, then let $t_j[i] := 1$ and let $t_j[h] := 0$ for all $h \in [n] \setminus \{ i \}$.
        \item
          If $\mathbf{S}'[h, j] = 0$ for some $h \in [n - 1]$,
          then $\mathbf{S}'[h', j] = \square$ for all $h' \in [n] \setminus \{ h \}$.
          If $h \ne i$, then let $t_j[i] := 1$ and let $t_j[h] := 0$ for all $h \in [n] \setminus \{ i \}$.
          Otherwise, let $t_j[h] := 1$ for all $h \in [n] \setminus \{ i \}$.
      \end{itemize}
    
    \item
      If $a_j^0 = 1$ and $a_j^1 = 1$, then add the edge $\{u_i,v_j\}$ if $\mathbf{S}[i, j] \neq\square$.
      \item
      If $a_j^0 \ge 2$ and $a_j^1 = 0$, then add the edge $\{u_i,v_j\}$ if $\mathbf{S}[i, j] = \square$.
    \item
      If $a_j^0 \ge 2$ and $a_j^1 = 1$, then add the edge $\{u_i,v_j\}$ if $\mathbf{S}[i, j] = 1$ (because $\mathbf{S}[n, j]$ must be completed with 0).
  \end{itemize}
  
  The correctness of the reduction easily follows from the definition of an edge:
  If~$\mathbf{T}$ is a solution for~$(\mathbf{S},s,m)$, then the corresponding subgraph of~$G$ contains the edge~$\{u_i,v_j\}$ for each~$i\in[n-1]$ and each~$j\in D(\mathbf{T}[i],\mathbf{T}[n])$.
  Conversely, a completion of~$\mathbf{S}$ is obtained from a subgraph~$G'$ by taking for each edge~$\{u_i,v_j\}$ the corresponding completion~$t_j$ as the~$j$-th column.
  Note that no vertex~$v_j$ can have two incident edges since~$f(v_j)=1$.
  Moreover, if~$v_j$ has no incident edges, then this implies that~$g(v_j)=a_j^1=0$. Hence, we can complete all missing entries in column~$j$ by~0.
  
  \begin{figure}[t]
    \centering
    \begin{subfigure}{.4\textwidth}
      \centering
      \begin{tikzpicture}[scale=0.6]
        \fill [lipicsLightGray] (0, 3) rectangle (1, 4);
        \fill [lipicsLightGray] (1, 1) rectangle (2, 2);
        \fill [lipicsLightGray] (1, 2) rectangle (2, 3);
        \fill [lipicsLightGray] (2, 2) rectangle (3, 3);
        \fill [lipicsLightGray] (2, 3) rectangle (3, 4);
        \fill [lipicsLightGray] (3, 0) rectangle (4, 1);
        \fill [lipicsLightGray] (3, 1) rectangle (4, 2);
        \fill [lipicsLightGray] (4, 1) rectangle (5, 2);
        \fill [lipicsLightGray] (4, 2) rectangle (5, 3);
        \fill [lipicsLightGray] (4, 4) rectangle (5, 5);

        \node at (0.5, 3.5) {0};
        \node at (1.5, 1.5) {0}; \node at (1.5, 2.5) {0};
        \node at (2.5, 2.5) {0}; \node at (2.5, 3.5) {1};
        \node at (3.5, 0.5) {0}; \node at (3.5, 1.5) {1};
        \node at (4.5, 1.5) {0}; \node at (4.5, 4.5) {0}; \node at (4.5, 2.5) {1};

        \node at (0.5, 0.5) {1}; \node at (0.5, 1.5) {1}; \node at (0.5, 2.5) {1}; \node at (0.5, 4.5) {1};
        \node at (1.5, 0.5) {0}; \node at (1.5, 3.5) {0}; \node at (1.5, 4.5) {1};
        \node at (2.5, 0.5) {0}; \node at (2.5, 1.5) {0}; \node at (2.5, 4.5) {0};
        \node at (3.5, 2.5) {0}; \node at (3.5, 3.5) {0}; \node at (3.5, 4.5) {0};
        \node at (4.5, 0.5) {0}; \node at (4.5, 3.5) {0};

        \draw[very thick] (0, 3) rectangle (1, 4);
        \draw[very thick] (1, 4) rectangle (2, 5);
        \draw[very thick] (2, 3) rectangle (3, 4);
        \draw[very thick] (3, 1) rectangle (4, 2);
        \draw[very thick] (4, 2) rectangle (5, 3);

        \node at (-0.5, 4.5) {$u_1$}; \node at (-0.5, 3.5) {$u_2$}; \node at (-0.5, 2.5) {$u_3$}; \node at (-0.5, 1.5) {$u_4$};
        \node at (0.5, 5.5) {$v_1$}; \node at (1.5, 5.5) {$v_2$}; \node at (2.5, 5.5) {$v_3$}; \node at (3.5, 5.5) {$v_4$}; \node at (4.5, 5.5) {$v_5$};
        \draw[help lines] (0, 0) grid (5, 5);
      \end{tikzpicture}
    \end{subfigure}
    \begin{subfigure}{.4\textwidth}
      \centering
      \begin{tikzpicture}[scale=0.7, every node/.style={circle,draw,inner sep=0pt,minimum size=7pt}]
        \node at (0.5, 2.5) (u1) [label=above:$u_1$] {};
        \node at (1.5, 2.5) (u2) [label=above:$u_2$] {};
        \node at (2.5, 2.5) (u3) [label=above:$u_3$] {};
        \node at (3.5, 2.5) (u4) [label=above:$u_4$] {};

        \node at (0.0, 0) (v1) [label=below:$v_1$] {};
        \node at (1.0, 0) (v2) [label=below:$v_2$] {};
        \node at (2.0, 0) (v3) [label=below:$v_3$] {};
        \node at (3.0, 0) (v4) [label=below:$v_4$] {};
        \node at (4.0, 0) (v5) [label=below:$v_5$] {};

        \draw (v1) -- (u1); \draw[very thick] (v1) -- (u2); \draw (v1) -- (u3); \draw (v1) -- (u4);
        \draw[very thick] (v2) -- (u1); \draw (v2) -- (u2);
        \draw[very thick] (v3) -- (u2); \draw (v3) -- (u3);
        \draw[very thick] (v4) -- (u4);
        \draw[very thick] (v5) -- (u3);
      \end{tikzpicture}
    \end{subfigure}
    \caption{A completion of a $5 \times 5$ incomplete matrix (left).
      The known entries are highlighted in gray.
      A bipartite graph as constructed in the reduction (right).
      Note that the entries framed by thick lines (which differ from all others in the same column) correspond to the subgraph represented by the thick lines.
    }
    \label{fig:factor}
  \end{figure}

  Regarding the running time,
  note that the constructed graph~$G$ has at most $n \ell$ edges and $\sum_{i \in [n - 1]} f(u_i) \in O(n)$ and $\sum_{j \in [\ell]} f(v_j) \in O(\ell)$.
  Since \textsc{$(g, f)$-Factor} can be solved in $O(|E| \sqrt{f(V)})$ time \cite{Gab83} for $f(V) = \sum_{v \in V} f(v)$, \textsc{SMC} can be solved in $O(n\ell \sqrt{n + \ell})$ time.
\end{proof}

Using \Cref{lemma:pmc}, we first show that $(\alpha, \alpha)$-\mc{} can be solved in polynomial time.

\begin{theorem}
  \label{theorem:aapoly}
  $(\alpha, \alpha)$-\mc{} can be solved in $O(n\ell \sqrt{n + \ell})$ time.
\end{theorem}
\begin{proof}
  We first show that $(\alpha, \alpha)$-\mc{} can easily be solved if $\alpha$ is odd.
  Consider row vectors $u, v, w \in \{ 0, 1 \}^\ell$ and let $U := D(u, v)$ and $W := D(v, w)$.
  Then, $\dist(u, v) + \dist(v, w) + \dist(w, u) = |U| + |W| + (|U| + |W| - 2|U \cap W|) = 2(|U| + |W| - |U \cap W|)$ and hence $\dist(u, v) + \dist(v, w) + \dist(w, u)$ is even.
  Thus, we can immediately answer \textbf{No} if $n \ge 3$.
  It is also easy to see that \mc{} can be solved in linear time if $n \le 2$.

  We henceforth assume that $\alpha$ is even.
	Eiben et al.~\cite[Theorem 34]{EGKOS21} provided a linear-time algorithm for $(0, \alpha)$-\mc{} with constant~$n$ (and arbitrary~$\alpha$) using reductions to integer linear programming (ILP).
  To ensure that each pairwise distance is at most $\alpha$, they express this property as a linear constraint.
  By simply adding an analogous constraint enforcing that each pairwise distance is at least~$\alpha$, it follows that $(\alpha, \alpha)$-\mc{} can also be solved in linear time for constant~$n$.
  So we can assume that $n \ge (\alpha / 2)^2 + (\alpha / 2) + 3$ (otherwise $(\alpha, \alpha)$-\mc{} can be solved in linear time).
  We claim that there is a completion $\mathbf{T}$ of $\mathbf{S}$ with $\mind(\mathbf{T}) = \maxd(\mathbf{T}) = \alpha$ if and only if the \textsc{SMC} instance $(\mathbf{S}', \alpha / 2, \alpha n / 2)$ is a \textbf{Yes}-instance for the matrix $\mathbf{S}' \in \{ 0, 1, \square \}^{(n + 1) \times \ell}$ obtained from $\mathbf{S}$ with an additional row vector $\square^\ell$.

  $(\Rightarrow)$
  Let $\mathbf{T}$ be a completion of $\mathbf{S}$ with $\mind(\mathbf{T}) = \maxd(\mathbf{T}) = \alpha$.
  Then, $\mathcal{T}$ is a weak $\Delta$-system with intersection size $\alpha / 2$:
  For any two sets $U, W \in \mathcal{T}$, we have $|U \cap W| = (|U| + |W| - |U \triangle W|) / 2 = \alpha / 2$.
  Since $|\mathcal{T}| \ge (\alpha / 2)^2 + (\alpha / 2) + 2$, \Cref{prop:sunflower} tells us that $\mathcal{T}$ is a sunflower.
  Let $C$ be the core of $\mathcal{T}$.
  Consider the completion $\mathbf{T}'$ of $\mathbf{S}'$ such that
  \begin{itemize}
    \item
      $\mathbf{T}'[[n], :] = \mathbf{T}$,
    \item
      $\mathbf{T}'[n + 1, j] = 1 - \mathbf{T}[n,j]$ for each $j \in C$, and
    \item
      $\mathbf{T}'[n + 1, j] = \mathbf{T}[n,j]$ for each $j \in [\ell] \setminus C$.
  \end{itemize}
  Note that $D(\mathbf{T}'[i], \mathbf{T}'[n + 1]) = D(\mathbf{T}'[i], \mathbf{T}'[n]) \setminus C$ for each $i \in [n - 1]$.
  Note also that $D(\mathbf{T}'[n], \mathbf{T}'[n + 1]) = C$.
  Hence, $D(\mathbf{T}'[1], \mathbf{T}'[n + 1]), \dots, D(\mathbf{T}'[n], \mathbf{T}'[n + 1])$ are pairwise disjoint sets of size $\alpha / 2$.

  $(\Leftarrow)$
  Let $\mathbf{T}'$ be a completion of $\mathbf{S}'$ such that $D(\mathbf{T}'[1], \mathbf{T}'[n + 1]), \dots, D(\mathbf{T}'[n], \mathbf{T}'[n + 1])$ are pairwise disjoint sets of size $\alpha / 2$.
  For the completion $\mathbf{T} = \mathbf{T}'[[n], :]$ of $\mathbf{S}$, it holds that $\dist(\mathbf{T}[i], \mathbf{T}[i']) = |D(\mathbf{T}'[i], \mathbf{T}'[n + 1]) \triangle D(\mathbf{T}'[i'], \mathbf{T}'[n + 1])| = |D(\mathbf{T}'[i], \mathbf{T}'[n + 1])| + |D(\mathbf{T}'[i'], \mathbf{T}'[n + 1])| = \alpha$ for each $i,i'\in[n]$.
\end{proof}

Now we proceed to develop polynomial-time algorithms for the case $\alpha + 1 = \beta$.
We will make use of the following observation made by Froese et al.~\cite[Proof of Theorem~9]{FBNS16}.

\begin{obs}
  \label{obs:intersectionsize}
  Let $\mathbf{T} \in \{ 0, 1 \}^{n \times \ell}$ with $\mind(\mathbf{T}) \ge \alpha$ and $\maxd(\mathbf{T}) \le \beta = \alpha + 1$.
  For $T_\alpha \ne T_\alpha' \in \mathcal{T}_\alpha$ and $T_\beta \ne T_\beta' \in \mathcal{T}_\beta$, it holds that $|T_\alpha \cap T_\alpha'| = \lfloor \alpha / 2 \rfloor$, $|T_\beta \cap T_\beta'| = \lceil \beta / 2 \rceil$, and $|T_\alpha \cap T_\beta| = \lceil \alpha / 2 \rceil = \lfloor \beta / 2 \rfloor$.
\end{obs}
\begin{proof}
  For any $T, T' \in \mathcal{T}_\alpha \cup \mathcal{T}_\beta$, we have $|T| + |T'| - 2 |T \cap T'| \in \{ \alpha, \beta \}$.
  If $T, T' \in \mathcal{T}_\alpha$ or $T, T' \in \mathcal{T}_\beta$, then $|T| + |T'| - 2 |T \cap T'|$ is even, and thus $|T| + |T'| - 2 |T \cap T'| \in 2 \lceil  \alpha / 2 \rceil = 2 \lfloor \beta / 2 \rfloor$.
  It follows that $|T_\alpha \cap T_\alpha'| = \lfloor \alpha / 2 \rfloor$ and $|T_\beta \cap T_\beta'| = \lceil \beta / 2 \rceil$.
  For the last equation, $|T_\alpha| + |T_\beta| - 2 |T_\alpha \cap T_\beta| = \alpha + \beta - 2 |T_\alpha \cap T_\beta| \in \{ \alpha, \beta \}$ gives us $2 |T_\alpha \cap T_\beta| \in \{ \alpha, \beta \}$.
  Hence, $|T_\alpha \cap T_\beta| = \lceil \alpha / 2 \rceil = \lfloor \beta / 2 \rfloor$.
\end{proof}

Surprisingly, an odd value of $\alpha$ seems to allow for significantly more efficient algorithms than an even value.

\begin{theorem}
  \label{theorem:aa+1}
  $(\alpha, \beta)$-\mc{} with $\beta=\alpha+1$ can be solved in
  \begin{enumerate}[label=(\roman*)]
    \item\label{aodd} $O(n\ell \sqrt{n + \ell})$ time for odd $\alpha$, and
    \item\label{aeven} $(n\ell)^{O(\alpha^3)}$ time for even~$\alpha$.
  \end{enumerate}
\end{theorem}
\begin{proof}
  \ref{aodd} We can assume that $n \ge \beta^2 / 2 + \beta + 7$ holds since otherwise the problem is linear-time solvable via a reduction to ILP as in the proof of \Cref{theorem:aapoly}.
  Suppose that $\mathbf{S}$ admits a completion $\mathbf{T}$ with $\mind(\mathbf{T}) \ge \alpha$ and $\maxd(\mathbf{T}) \le \beta$.
  Since $\mathcal{T} = \mathcal{T}_\alpha \cup \mathcal{T}_\beta$ and $|\mathcal{T}| \ge \beta^2 / 2 + \beta + 6$, it follows that $\max \{|\mathcal{T}_\alpha|, |\mathcal{T}_\beta| \} \ge c := (\beta / 2)^2 + (\beta / 2) + 3$.
  We consider two cases depending on the size of $\mathcal{T}_\alpha$ and~$\mathcal{T}_\beta$.
  \begin{enumerate}
    \item
      Suppose that $|\mathcal{T}_\alpha| \ge c$.
      Since $\mathcal{T}_\alpha$ is a weak $\Delta$-system with intersection size $(\alpha - 1) / 2$,
      $\mathcal{T}_\alpha$ is a sunflower with a core of size $(\alpha - 1) / 2$ and petals of size $(\alpha + 1) / 2$ by \Cref{prop:sunflowerodd}~\ref{mu}.
      We claim that $\mathcal{T}_\beta = \emptyset$.
      Suppose not and let $T_\beta \in \mathcal{T}_\beta$.
      We then obtain $|T_\alpha \cap T_\beta| = (\alpha + 1) / 2$ for all $T_\alpha \in \mathcal{T}_\alpha$ by \Cref{obs:intersectionsize}, which contradicts \Cref{lemma:containcore}.
    \item
      Suppose that $|\mathcal{T}_\beta| \ge c$.
      By \Cref{prop:sunflower}, $\mathcal{T}_\beta$ is a sunflower whose core $C$ has size $\beta / 2$.
      By \Cref{obs:intersectionsize} and \Cref{lemma:containcore}, $T_\alpha \supseteq C$ holds for each $T_\alpha \in \mathcal{T}_\alpha$.
      Now suppose that there exist $T_\alpha \ne T_\alpha' \in \mathcal{T}_\alpha$.
      Since $C \subseteq T_\alpha$ and $C \subseteq T_\alpha'$, it follows that $|T_\alpha \cap T_\alpha'| \ge \beta / 2$, thereby contradicting \Cref{obs:intersectionsize}.
      Hence, we have $|\mathcal{T}_\alpha| \le 1$.
  \end{enumerate}
  We construct an instance $I$ of \textsc{SMC} covering both cases above, as in \Cref{theorem:aapoly}.
  We use the matrix~$\mathbf{S}'$ obtained from~$\mathbf{S}$ by appending a row vector $\square^\ell$, and we set $s := \beta / 2$ and $m := ns - 1$.
  Basically, we allow at most one ``petal'' to have size~$s-1$.
  We return \textbf{Yes} if and only if $I$ is a \textbf{Yes}-instance.
  The correctness can be shown analogously to the proof of \Cref{theorem:aapoly}.
  
  \ref{aeven} Suppose that there is a completion $\mathbf{T}$ of $\mathbf{S}$ with $\mind(\mathbf{T}) \ge \alpha$ and $\maxd(\mathbf{T}) \le \beta$.
  Again, we can assume that $n > 2c$ for $c := (\beta / 2)^2 + (\beta / 2) + 4$, and consider a case distinction regarding the size of $\mathcal{T}_\alpha$ and $\mathcal{T}_\beta$.
  \begin{enumerate}
    \item
      Suppose that $|\mathcal{T}_\alpha| \ge c$ and $|\mathcal{T}_\beta| \ge c$.
      It follows from \Cref{obs:intersectionsize,prop:sunflower,prop:sunflowerodd} that $\mathcal{T}_\alpha$ and $\mathcal{T}_\beta$ are sunflowers.
      Let $C_\alpha$ and $C_\beta$ be the cores of $\mathcal{T}_\alpha$ and $\mathcal{T}_\beta$, respectively.
      Note that $|C_\alpha| = \alpha / 2$ and $|C_\beta| = \alpha / 2 + 1$, and hence $C_\alpha \subsetneq C_\beta$ holds by \Cref{obs:intersectionsize,lemma:containcore}.
      Let $j \in [\ell]$ be such that $C_\alpha \cup \{ j \} = C_\beta$ and let $\mathbf{T}' := \mathbf{T}[:, [\ell] \setminus \{ j \}]$.
      Then, the set family $\mathcal{T}'$ is a sunflower with a core of size $\alpha / 2$ and petals of size $\alpha / 2$.
      Hence, there exists $j \in [\ell]$ such that the $(\alpha, \alpha)$-\mc{} instance $\mathbf{S}[:, [\ell] \setminus \{ j \}]$ is a \textbf{Yes}-instance.
      On the other hand, if there is a completion $\mathbf{T}'$ of $\mathbf{S}[:, [\ell] \setminus \{ j \}]$ with $\mind(\mathbf{T}') = \maxd(\mathbf{T}') = \alpha$, then $\mind(\mathbf{T}) \ge \alpha$ and $\maxd(\mathbf{T}) \le \alpha + 1$ hold for any completion $\mathbf{T}$ of $\mathbf{S}$ with $\mathbf{T}[:, [\ell] \setminus \{ j \}]=\mathbf{T}'$.
    \item
      Suppose that $|\mathcal{T}_\alpha| \ge c$ and $|\mathcal{T}_\beta| < c$.
      The same argument as above shows that $T_\alpha \cap T_\beta = C$ holds for each $T_\alpha \in \mathcal{T}_\alpha$ and $T_\beta \in \mathcal{T}_\beta$, where $C$ is the size-$\alpha/2$ core of sunflower $\mathcal{T}_\alpha$.
      Let $I_\beta = \{ i \in [n - 1] \mid \dist(\mathbf{T}[i], \mathbf{T}[n]) = \beta \}$ be the row indices that induce the sets in $\mathcal{T}_\beta$ and let $J_\beta = \bigcup_{T_\beta \in \mathcal{T}_\beta} T_\beta$.
      Consider $\mathbf{T}' = \mathbf{S}[[n] \setminus I_\beta, [\ell] \setminus (C \cup J_\beta)]$
      and note that the family $\mathcal{T}'$ consists of pairwise disjoint sets, each of size $\alpha / 2$.
      We use this observation to obtain a reduction to \textsc{SMC}.
      The idea is to test all possible choices for~$\mathbf{T}'$, that is, we simply try out all possibilities to choose the following sets:
      \begin{itemize}
        \item
          $C \subseteq [\ell]$ of size exactly $\alpha / 2$.
        \item
          $I_\beta \subseteq [n - 1]$ of size at most $c$.
        \item
          $J_\beta \subseteq [\ell] \setminus C$ of size at most $\beta \cdot c$ such that $\dist_{[\ell] \setminus (C \cup J_\beta)}(\mathbf{S}[i_\beta], \mathbf{S}[n]) = 0$ for all $i_\beta \in I_\beta$.
        \end{itemize}
        For each possible choice, we check whether it allows for a valid completion.
        Formally, it is necessary that the following exist:
        \begin{itemize}
        \item
          A completion $t_C$ of $\mathbf{S}[n, C]$ such that $\mathbf{S}[i, j] \ne t_C[j]$ for all $i \in [n - 1]$ and $j \in C$.
          \item
          A completion $t_{J_\beta}$ of $\mathbf{S}[n, J_\beta]$ such that $\dist(t_{J_\beta}, \mathbf{S}[i, J_\beta]) = 0$ for all $i \in [n - 1] \setminus I_\beta$.
        \item
          A completion $t_{i_\beta}$ of $\mathbf{S}[i_\beta, J_\beta]$ for each $i_\beta \in I_\beta$ such that $\dist(t_{i_\beta}, t_{J_\beta}) = \alpha / 2 + 1$ for each $i_\beta \in I_\beta$ and $\dist(t_{i_\beta}, t_{i_\beta'}) = \alpha$ for each $i_\beta \ne i_\beta' \in I_\beta$.
          
        \end{itemize}
        The existence of the above completions can be checked in~$O(n)$ time.
      We then construct an \textsc{SMC} instance $(\mathbf{S}', \alpha / 2, (n - |I_\beta| - 1) \cdot \alpha / 2)$, where $\mathbf{S}'$ is an incomplete matrix with $n' = n - |I_\beta|$ rows and $\ell - |C| - |J_\beta|$ columns defined as follows:
      \begin{itemize}
        \item
          $\mathbf{S}'[[n' - 1]] = \mathbf{S}[[n - 1] \setminus I_\beta, [\ell] \setminus (C \cup J_\beta)]$.
        \item
          $\mathbf{S}'[n', j] = \square$ for each $j \in [\ell] \setminus (C \cup J_\beta)$ such that $\mathbf{S}[i_\beta, j] = \square$ for all $i_\beta \in I_\beta \cup \{ n \}$.
        \item
          $\mathbf{S}'[n', j] = \mathbf{S}[i_\beta, j]$ for each $j \in [\ell] \setminus (C \cup J_\beta)$ such that $\mathbf{S}[i_\beta, j] \neq \square$ for some $i_\beta \in I_\beta \cup \{ n \}$.
      \end{itemize}
      Overall, we solve at most $(n\ell)^{O(\alpha^3)}$ \textsc{SMC} instances and hence it requires~$(n\ell)^{O(\alpha^3)}$ time.
    \item
      Suppose that $|\mathcal{T}_\alpha| < c$ and $|\mathcal{T}_\beta| \ge c$.
      Let $i \in [n - 1]$ be such that $\dist(\mathbf{T}[i], \mathbf{T}[n]) = \beta$.
      Then, $\dist(\mathbf{T}[i], \mathbf{T}[i']) = \alpha$ holds for each $i' \in [n-1] \setminus \{ i \}$ with $\dist(\mathbf{T}[i'], \mathbf{T}[n]) = \beta$.
      Since there are at least $c - 1 = (\beta / 2)^2 + (\beta / 2) + 3$ such row indices, it follows that this case is essentially equivalent to the previous case (by considering row~$i$ as the last row).
  \end{enumerate}
\end{proof}

A natural question is whether one can extend our approach above to the case $\beta = \alpha + 2$ (particularly $\alpha=1$ and $\beta=3$).
The problem is that the petals of the sunflowers $\mathcal{T}_2$ and~$\mathcal{T}_3$ may have nonempty intersections.
Thus, reducing to \textsc{SMC} to obtain a polynomial-time algorithm is probably hopeless.

\subsection{NP-hardness}
\label{sec:const:nph}

Hermelin and Rozenberg~\cite[Theorem 5]{HR15} proved that \textsc{ConRMC} (under the name \textsc{Closest String with Wildcards}) is NP-hard even if $r[i] = 2$ for all $i \in [n]$.
Using this result, we prove the following.

\begin{restatable}{theorem}{dfournphard}
  \label{theorem:d4nphard}
  $(\alpha, \beta)$-\mc{} is NP-hard if $\beta \ge 2 \lceil \alpha / 2 \rceil + 4$.
\end{restatable}
\begin{proof}
  We give a polynomial-time reduction from~\textsc{ConRMC}.
  Let $(\mathbf{S} \in \{ 0, 1, \square \}^{n \times \ell},r$) be a \textsc{ConRMC} instance with~$r[i]=2$ for all~$i\in[n]$.
  
  Let $\mathbf{C}\in\{0,1\}^{(n+1)\times m}$ be the binary matrix with $m :=(n-1) \cdot \lceil \alpha / 2 \rceil + \beta - 2$ columns obtained by horizontally stacking
  \begin{itemize}
    \item
      the $(n + 1) \times (n + 1)$ identity matrix $\lceil \alpha / 2 \rceil$ times and
    \item
      the column vector $(0^n 1)^T$ $\beta - 2 \lceil \alpha / 2 \rceil - 2$ times.
  \end{itemize}
  Since the pairwise row Hamming distances in the identity matrix are all two, we have that:
  \begin{itemize}
    \item 
      $\dist(\mathbf{C}[i], \mathbf{C}[i']) = 2 \lceil \alpha / 2 \rceil$ for each $i \ne i' \in [n]$ and
    \item
      $\dist(\mathbf{C}[i], \mathbf{C}[n + 1]) = \beta - 2$ for each $i \in [n]$.
  \end{itemize}
  Consider the matrix $\mathbf{S}'\in\{ 0, 1, \square \}^{(n+1)\times(\ell+m)}$ obtained from $\mathbf{S}$ by adding the row $\square^{\ell}$ and then horizontally appending $\mathbf{C}$.
  We show that there exists a vector $v \in \{ 0, 1 \}^{\ell}$ with $\dist(v, \mathbf{S}[i]) \le 2$ for all $i \in [n]$ if and only if $\mathbf{S}'$ admits a completion $\mathbf{T}'$ with $\mind(\mathbf{T}') \ge \alpha$ and $\maxd(\mathbf{T}') \le \beta$.

  $(\Rightarrow)$
  Let $\mathbf{T}'$ be the completion of $\mathbf{S}'$ such that for each $i \in [n + 1]$, $\mathbf{T}'[i, [\ell]] = \mathbf{S}'[i] \leftarrow v$.
  Then, we have the following:
  \begin{itemize}
    \item
      $\mind(\mathbf{T'}) \ge \mind(\mathbf{C}) = 2 \lceil \alpha / 2 \rceil \ge \alpha$.
    \item
      $\dist(\mathbf{T}'[i], \mathbf{T}'[n + 1]) = \dist(\mathbf{S}[i], v) + \beta - 2 \le \beta$ for each $i \in [n]$.
    \item
      By the triangle inequality, $\dist(\mathbf{T}'[i], \mathbf{T}'[i']) \le \dist(v, \mathbf{S}[i]) + \dist(v, \mathbf{S}[i']) + \dist(\mathbf{C}[i], \mathbf{C}[i']) \le 2 + 2 + 2 \lceil \alpha / 2 \rceil \le \beta$ holds for each $i, i' \in [n]$ 
    \end{itemize}

  \indent $(\Leftarrow)$
  Let $v = \mathbf{T'}[n + 1, [\ell]]$.
  It is easy to see that $\dist(\mathbf{S}[i], v) = \dist(\mathbf{T}'[i], \mathbf{T}'[n + 1]) - \dist(\mathbf{C}[i], \mathbf{C}[n + 1]) \le \beta - (\beta - 2) \le 2$ holds for each $i \in [n]$.
\end{proof}

It remains open whether NP-hardness also holds for~$(\alpha,\alpha+3)$-\mc{} with~$\alpha \ge 1$ (recall that $(0,3)$-\mc{} is polynomial-time solvable).
In \Cref{sec:mindmc:k2}, however, we show NP-hardness for~$\beta = \alpha + 3$ when $\alpha$ and $\beta$ are part of the input.

\section{Bounded Number \texorpdfstring{$k$}{k} of Missing Entries per Row}
\label{sec:mindmc:k2}

In this section, we consider \mc{} with $\alpha$ and $\beta$ being part of the input, hence not necessarily being constants.
We consider the maximum number $k$ of missing entries in any row as a parameter (\mc{} is clearly trivial for $k = 0$).
We obtain two polynomial-time algorithms and two NP-hardness results.
Our polynomial-time algorithms are based on reductions to \textsc{2-SAT}.

\subsection{Polynomial-time algorithms}
\label{sec:unbounded:poly}

We show that \mc{} can be solved in polynomial time when $k = 1$, via a reduction to \textsc{2-SAT}.
For a Boolean variable~$x$, we use $(x = 1)$ and $(x \ne 0)$ to denote the positive literal~$x$.
Similarly, we use $(x = 0)$ and $(x \ne 1)$ for the negative literal $\neg x$.

\begin{restatable}{theorem}{kpoly}
  \label{theorem:k1polynomial}
  \mc{} can be solved in $O(n^2 \ell)$ time 
  \begin{enumerate}[label=(\roman*)]
    \item \label{kone} for $k = 1$, and
    \item \label{ktwoaa} for $k=2$ and $\alpha = \beta$.
  \end{enumerate}

\end{restatable}

\begin{proof}
  \ref{kone} We construct a 2-CNF formula $\phi$ of polynomial size such that $\phi$ is satisfiable if and only if the input matrix $\mathbf{S}$ admits a completion $\mathbf{T}$ with $\mind(\mathbf{T}) \ge \alpha$ and $\maxd(\mathbf{T}) \le \beta$.

  First, we compute the distances $\dist(\mathbf{S}[i], \mathbf{S}[i'])$ for each $i, i' \in [n]$ in $O(n^2 \ell)$ time.
  Clearly, if there exists a pair with distance less than $\alpha - 2$ or larger than $\beta$, then we have a \textbf{No}-instance.
  Let $I \subseteq [n]$ be the set of row indices corresponding to row vectors with a missing entry and let $j_i \in [\ell]$ be such that $\mathbf{S}[i, j_i] = \square$.

  We introduce a variable $x_i$ for each $i \in I$, where $x_{i}$ is set to true if $\mathbf{S}[i,j_i]$ is completed with a 1. We construct the formula $\phi$ as follows:
  \begin{itemize}
    \item
      For each $i < i' \in [n]$ with $\dist(\mathbf{S}[i], \mathbf{S}[i']) = \alpha - 2$:
      
          If $i \not\in I$ or $i' \not\in I$, or $j_i = j_{i'}$, then return \textbf{No}.
          Otherwise, add the clauses $(x_i = 1 - \mathbf{S}[i', j_i])$ and $(x_{i'} = 1 - \mathbf{S}[i, j_{i'}])$.
     
    \item
      For each $i < i' \in [n]$ with $\dist(\mathbf{S}[i], \mathbf{S}[i']) = \alpha - 1$:
      \begin{itemize}
        \item
          If $i \not\in I$ and $i' \not\in I$, then return \textbf{No}.
        \item
          If $i \in I$ and $i' \not\in I$, then add the clause $(x_i \ne \mathbf{S}[i', j_i])$.
        \item
          If $i \not\in I$ and $i' \in I$, then add the clause $(x_{i'} \ne \mathbf{S}[i, j_{i}])$.
        \item
          If $i\in I$ and~$i'\in I$ and $j_i = j_{i'}$, then add the clauses $(x_i \vee x_{i'})$ and $(\neg x_i \vee \neg x_{i'})$.
        \item
          If $i\in I$ and~$i'\in I$ and $j_i \neq j_{i'}$, then add the clause $(x_{i} \ne \mathbf{S}[i', j_i] \vee x_{i'} \ne \mathbf{S}[i, j_{i'}])$.
      \end{itemize}
    \end{itemize}
    It is easy to see that these clauses ensure that $\mind(\mathbf{T})\ge \alpha$.
    Similarly, to ensure that $\maxd(\mathbf{T})\le \beta$, we add the following clauses:
    \begin{itemize}
    \item
      For each $i < i' \in [n]$ with $\dist(\mathbf{S}[i], \mathbf{S}[i']) = \beta$:
      \begin{itemize}
          \item If $i \in I$ and $i' \not\in I$, then add the clause~$(x_i=\mathbf{S}[i', j_{i}])$.
          \item If $i \not\in I$ and $i' \in I$, then add the clause~$(x_{i'}=\mathbf{S}[i, j_{i'}])$.
          \item If $i \in I$ and $i'\in I$ and $j_i = j_{i'}$, then add the clauses~$(x_i\vee \neg x_{i'})$ and $(\neg x_i \vee x_{i'})$.
            \item If $i \in I$ and $i'\in I$ and $j_i \neq j_{i'}$, then add the clauses $(x_i=\mathbf{S}[i', j_{i}])$ and $(x_{i'}=\mathbf{S}[i, j_{i'}])$.
          \end{itemize}
    \item
      For each $i < i' \in [n]$ with $\dist(\mathbf{S}[i], \mathbf{S}[i']) = \beta - 1$:
      
          If $i\in I$ and~$i'\in I$ and $j_i \neq j_{i'}$, then add the clause $(x_{i} = \mathbf{S}[i', j_i] \vee x_{i'} = \mathbf{S}[i, j_{i'}])$.
        \end{itemize}
        Thus,~$\phi$ is of size~$O(n^2)$ and can be solved in $O(n^2)$ time~\cite{aspvall1979linear}.
        The correctness follows directly from the construction.

  \ref{ktwoaa}
  Since~$k=2$, there are at most four possible ways to complete the row vector $\mathbf{S}[n]$.
  So we can afford to try all possible completions of $\mathbf{S}[n]$.
  Without loss of generality, assume that $\mathbf{S}[n] = 0^\ell$.
  We show that \mc{} can be solved in $O(n^2 \ell)$ time. 
  
  First, we check whether $\alpha - 2 \le \dist(\mathbf{S}[i], 0^\ell) \le \alpha$ holds for each $i \in [n - 1]$ (otherwise, we return \textbf{No}).
  We do the following for each $i \in [n - 1]$:
  \begin{itemize}
    \item
      $\mathbf{S}[i]$ contains exactly one missing entry:
      If $\dist(\mathbf{S}[i], 0^\ell) = \alpha - 1$, then fill the missing entry by~1. 
      If $\dist(\mathbf{S}[i], 0^\ell) = \alpha$, then fill the missing entry by~0.
      If $\dist(\mathbf{S}[i], 0^\ell) = \alpha - 2$, then return \textbf{No}.
    \item
      $\mathbf{S}[i]$ contains exactly two missing entries:
      If $\dist(\mathbf{S}[i], 0^\ell) = \alpha - 2$, then fill both missing entries by~1.
      If $\dist(\mathbf{S}[i], 0^\ell) = \alpha$, then fill both missing entries by~0.
  \end{itemize}
  Now, each row vector either contains no missing entry or exactly two missing entries.
  Let $I_0 \subseteq [n]$ be the set of row indices corresponding to row vectors without any missing entry and let $I_2 = [n] \setminus I_0$.
  We check in~$O(n^2 \ell)$ time whether all pairwise Hamming distances in $\mathbf{S}[I_0]$ are $\alpha$. If not, then we return \textbf{No}.
  Note that we have $\dist(\mathbf{S}[i], 0^\ell) = \alpha - 1$ for each $i \in I_2$, and thus there are exactly two ways to complete $\mathbf{S}[i]$:
  One missing entry filled by~1 and the other by~0, or vice versa.
  For each $i \in I_2$, let $j_i^1 < j_i^2 \in [\ell]$ be such that $\mathbf{S}[i, j_i^1] = \mathbf{S}[i, j_i^2] = \square$.
  We verify whether the following necessary conditions hold:
  \begin{itemize}
    \item
      $\dist(\mathbf{S}[i_0], \mathbf{S}[i]) = \alpha - 1$ for each $i_0 \in I_0$ and $i \in I_2$ with $\mathbf{S}[i_0, j_i^1] = \mathbf{S}[i_0, j_i^2]$.
    \item
      $\dist(\mathbf{S}[i_0], \mathbf{S}[i]) \in  \{ \alpha- 2, \alpha \}$ for each $i_0 \in I_0$ and $i \in I_2$ with $\mathbf{S}[i_0, j_i^1] \ne \mathbf{S}[i_0, j_i^2]$.
    \item
      Let $i \ne i' \in I_2$ be such that $\{ j_i^1, j_i^2 \} \cap \{ j_{i'}^1, j_{i'}^2 \} = \emptyset$.
      Observe that if $\mathbf{S}[i', j_{i}^1] = \mathbf{S}[i', j_{i}^2]$, then the completion of $\mathbf{S}[i]$ increases the distance between $\mathbf{S}[i]$ and $\mathbf{S}[i']$ by exactly one.
      Otherwise, the distance either stays the same or increases by exactly two.
      It is analogous for the completion of $\mathbf{S}[i']$.
      Thus, the following must hold for each $i \ne i' \in I_2$ with $\{ j_i^1, j_i^2 \} \cap \{ j_{i'}^1, j_{i'}^2 \} = \emptyset$:
      \begin{itemize}
        \item
          $\dist(\mathbf{S}[i], \mathbf{S}[i']) = \alpha - 2$ if $\mathbf{S}[i, j_{i'}^1] = \mathbf{S}[i, j_{i'}^2]$ and $\mathbf{S}[i', j_{i}^1] = \mathbf{S}[i', j_{i}^2]$.
        \item
          $\dist(\mathbf{S}[i], \mathbf{S}[i']) = \{ \alpha - 3, \alpha - 1\}$ if either $\mathbf{S}[i, j_{i'}^1] = \mathbf{S}[i, j_{i'}^2]$ and $\mathbf{S}[i', j_{i}^1] \ne \mathbf{S}[i', j_{i}^2]$, or $\mathbf{S}[i, j_{i'}^1] \ne \mathbf{S}[i, j_{i'}^2]$ and $\mathbf{S}[i', j_{i}^1] = \mathbf{S}[i', j_{i}^2]$.
        \item
          $\dist(\mathbf{S}[i], \mathbf{S}[i']) = \{ \alpha - 4, \alpha - 2, \alpha \}$ if $\mathbf{S}[i, j_{i'}^1] \ne \mathbf{S}[i, j_{i'}^2]$ and $\mathbf{S}[i', j_{i}^1] \ne \mathbf{S}[i', j_{i}^2]$.
      \end{itemize}
    \item
      Let $i \ne i' \in I_2$ be such that $j_i^2 = j_{i'}^1$.
      Note that $\mathbf{S}[i, j_{i}^2]$ and $\mathbf{S}[i', j_{i'}^  1]$ are completed by the same value, if and only if $\mathbf{S}[i, j_i^{1}]$ and $\mathbf{S}[i', j_{i'}^{2}]$ are completed by the same value.
      Hence, the following must hold for each $i \ne i' \in I_2$ with $j_i^2 = j_{i'}^1$:
      \begin{itemize}
        \item
          $\dist(\mathbf{S}[i], \mathbf{S}[i']) \in \{ \alpha - 2, \alpha \}$ if $\mathbf{S}[i, j_{i'}^2] = \mathbf{S}[i', j_{i}^1]$.
        \item
          $\dist(\mathbf{S}[i], \mathbf{S}[i']) \in \{ \alpha - 3, \alpha - 1 \}$ if $\mathbf{S}[i, j_{i'}^2] \ne \mathbf{S}[i', j_{i}^1]$.
      \end{itemize}
    \item
      $\dist(\mathbf{S}[i], \mathbf{S}[i']) \in \{ \alpha - 2, \alpha \}$ for each $i \ne i' \in I_2$ with $j_i^1 = j_{i'}^1$ and $j_{i}^2 = j_{i'}^2$.
  \end{itemize}
  We return \textbf{No} if at least one of the above fails.
  Clearly this requires $O(n^2 \ell)$ time.

  Now, we construct a 2-CNF formula which is satisfiable if and only if our \mc{} instance is a \textbf{Yes}-instance.
  We introduce a variable $x_i$ for each $i \in I_2$, which basically encodes the completion of $\mathbf{S}[i]$.
  Intuitively speaking, $(\mathbf{S}[i, j_i^1], \mathbf{S}[i, j_i^2])$ are filled by $(1, 0)$ if $x_i$ is true and by $(0, 1)$ if $x_i$ is false.
  We add clauses as follows:
  \begin{itemize}
    \item
      For each $i_0 \in I_0$ and $i \in I_2$ with $\dist(\mathbf{S}[i_0], \mathbf{S}[i]) = \alpha - 2$, add a singleton clause $(x_i)$ if $(\mathbf{S}[i_0, j_i^1], \mathbf{S}[i_0, j_i^2]) = (0, 1)$ and $(\neg x_i)$ otherwise.
    \item
      For each $i \ne i' \in I_2$ with $\{ j_i^1, j_i^2 \} \cap \{ j_{i'}^1, j_{i'}^2 \} = \emptyset$, $\mathbf{S}[i, j_{i'}^1] = \mathbf{S}[i, j_{i'}^2]$, and $\mathbf{S}[i, j_{i'}^1] \ne \mathbf{S}[i, j_{i'}^2]$, add a clause $(x_i)$ if either $(\mathbf{S}[i, j_i^1], \mathbf{S}[i, j_i^{2}]) = (0, 1)$ and $\dist(\mathbf{S}[i], \mathbf{S}[i']) = \alpha - 3$, or $(\mathbf{S}[i, j_i^1], \mathbf{S}[i, j_i^{2}]) = (1, 0)$ and $\dist(\mathbf{S}[i], \mathbf{S}[i']) = \alpha - 1$.
      Analogously, we add a clause $(x_{i'})$ or $(\neg x_{i'})$ for the case $\{ j_i^1, j_i^2 \} \cap \{ j_{i'}^1, j_{i'}^2 \} = \emptyset$, $\mathbf{S}[i, j_{i'}^1] = \mathbf{S}[i, j_{i'}^2]$, and $\mathbf{S}[i, j_{i'}^1] \ne \mathbf{S}[i, j_{i'}^2]$.
    \item
      We do as follows for each $i \ne i' \in I_2$ with $\{ j_i^1, j_i^2 \} \cap \{ j_{i'}^1, j_{i'}^2 \} = \emptyset$, $\mathbf{S}[i, j_{i'}^1] \ne \mathbf{S}[i, j_{i'}^2]$, and $\mathbf{S}[i, j_{i'}^1] \ne \mathbf{S}[i, j_{i'}^2]$.
      \begin{itemize}
        \item
          Suppose that $\dist(\mathbf{S}[i], \mathbf{S}[i']) = \alpha - 4$.
          Add a clause $(x_i)$ if $(\mathbf{S}[i', j_{i}^1], \allowbreak \mathbf{S}[i', j_{i}^2]) = (0, 1)$ and $(\neg x_i)$ otherwise.
          Moreover, add a clause $(x_{i'})$ if $(\mathbf{S}[i, j_{i'}^1], \allowbreak \mathbf{S}[i, j_{i'}^2]) = (0, 1)$ and $(\neg x_{i'})$ otherwise.
        \item
          Suppose that $\dist(\mathbf{S}[i], \mathbf{S}[i']) = \alpha - 2$.
          If $(\mathbf{S}[i, j_{i'}^1], \mathbf{S}[i, j_{i'}^2]) = (\mathbf{S}[i', j_i^1], \mathbf{S}[i', j_i^2])$, then add clauses $(x_i \vee x_{i'})$ and $(\neg x_i \vee \neg x_{i'})$.
          Note that these clauses are satisfied if and only if $x_i \ne x_{i'}$.
          Otherwise, we add clauses $(x_i \vee \neg x_{i'})$ and $(\neg x_i \vee x_{i'})$, which are satisfied if and only if $x_i = x_{i'}$.
        \item
          Suppose that $\dist(\mathbf{S}[i], \mathbf{S}[i']) = \alpha$.
          Add a clause $(x_i)$ if $(\mathbf{S}[i', j_{i}^1], \allowbreak \mathbf{S}[i', j_{i}^2]) = (1, 0)$ and $(\neg x_i)$ otherwise.
          Moreover, add a clause $(x_{i'})$ if $(\mathbf{S}[i, j_{i'}^1], \allowbreak \mathbf{S}[i, j_{i'}^2]) = (1, 0)$ and $(\neg x_{i'})$ otherwise.
      \end{itemize}
      \item
        We do as follows for each $i \ne i' \in I_2$ with $j_i^2 = j_{i'}^1$.
        \begin{itemize}
          \item
            If $\dist(\mathbf{S}[i], \mathbf{S}[i']) = \alpha - 2$ and $\mathbf{S}[i, j_{i'}^2] = \mathbf{S}[i', j_{i}^1] = 0$, then add a clause $(x_i \vee \neg x_{i'})$.
          \item
            If $\dist(\mathbf{S}[i], \mathbf{S}[i']) = \alpha$ and $\mathbf{S}[i, j_{i'}^2] = \mathbf{S}[i', j_{i}^1] = 0$, then add clauses $(\neg x_i)$ and $(x_{i'})$.
          \item
            If $\dist(\mathbf{S}[i], \mathbf{S}[i']) = \alpha - 2$ and $\mathbf{S}[i, j_{i'}^2] = \mathbf{S}[i', j_{i}^1] = 1$, then add a clause $(\neg x_i \vee x_{i'})$.
          \item
            If $\dist(\mathbf{S}[i], \mathbf{S}[i']) = \alpha$ and $\mathbf{S}[i, j_{i'}^2] = \mathbf{S}[i', j_{i}^1] = 1$, then add clauses $(x_i)$ and $(\neg x_{i'})$.
          \item
            If $\dist(\mathbf{S}[i], \mathbf{S}[i']) = \alpha - 3$ and $(\mathbf{S}[i, j_{i'}^2], \mathbf{S}[i', j_{i}^1]) = (1, 0)$, then add clauses $(x_i)$ and $(x_{i'})$.
          \item
            If $\dist(\mathbf{S}[i], \mathbf{S}[i']) = \alpha - 1$ and $(\mathbf{S}[i, j_{i'}^2], \mathbf{S}[i', j_{i}^1]) = (1, 0)$, then add a clause $(\neg x_i \vee \neg x_{i'})$.
          \item
            If $\dist(\mathbf{S}[i], \mathbf{S}[i']) = \alpha - 3$ and $(\mathbf{S}[i, j_{i'}^2], \mathbf{S}[i', j_{i}^1]) = (0, 1)$, then add clauses $(\neg x_i)$ and $(\neg x_{i'})$.
          \item
            If $\dist(\mathbf{S}[i], \mathbf{S}[i']) = \alpha - 1$ and $(\mathbf{S}[i, j_{i'}^2], \mathbf{S}[i', j_{i}^1]) = (0, 1)$, then add a clause $(x_i \vee x_{i'})$.
        \end{itemize}
        \item
          For each $i \ne i' \in I_2$ with $j_i^1 = j_{i'}^1$ and $j_{i}^2 = j_{i'}^2$,
          add clauses $(x_i \vee \neg x_{i'})$ and $(\neg x_i \vee x_{i'})$ if $\dist(\mathbf{S}[i], \mathbf{S}[i']) = \alpha - 2$ and
          add clauses $(x_i \vee x_{i'})$ and $(\neg x_i \vee \neg x_{i'})$ if $\dist(\mathbf{S}[i], \mathbf{S}[i']) = \alpha$.
  \end{itemize}

  It is easy to check that the constructed formula is correct.
  The formula contains $O(n^2)$ clauses and can thus be solved in~$O(n^2)$ time~\cite{aspvall1979linear}.
\end{proof}

Next, we show that the quadratic dependency on $n$ in the running time of \Cref{theorem:k1polynomial} is inevitable under the Orthogonal Vectors Conjecture (OVC), which states that \textsc{Orthogonal Vectors} cannot be solved in $O(n^{2 - \epsilon} \cdot \ell^c)$ time for any $\epsilon, c > 0$ (it is known that the Strong Exponential Time Hypothesis implies the OVC \cite{Wil05}).

\dprob{Orthogonal Vectors}
{Sets $\mathcal{U}, \mathcal{V}$ of row vectors in \( \{ 0, 1\}^\ell \) with \( |\mathcal{U}| = |\mathcal{V}| = n \).}
{Are there row vectors $u \in \mathcal{U}$ and $v \in \mathcal{V}$ such that \( u[j] \cdot v[j] = 0 \) holds for all $j \in [\ell]$?}

\begin{theorem}
  \mc{} cannot be solved in $O(n^{2 - \epsilon} \cdot \ell^c)$ time for any $c, \epsilon > 0$, unless OVC breaks.
\end{theorem}
\begin{proof}
  We reduce from \textsc{Orthogonal Vectors}.
  Let $u_1, \dots, u_n, v_1, \dots, v_n \in \{ 0, 1 \}^\ell$ be row vectors.
  Consider the matrix $\mathbf{T} \in \{ 0, 1 \}^{2n \times 6\ell}$ where
  \begin{align*}
    \mathbf{T}[i, [3j - 2, 3j]] &= \begin{cases}
      001 & \text{if } u_i[j] = 0, \\
      111 & \text{if } u_i[j] = 1, \\
    \end{cases} & 
    \mathbf{T}[i, [3 \ell + 3j - 2, 3 \ell + 3j]] = 000,
    \\
    \mathbf{T}[n + i, [3j - 2, 3j]] &= \begin{cases}
      010 & \text{if } v_i[j] = 0, \\
      111 & \text{if } v_i[j] = 1, \\
    \end{cases} &
    \mathbf{T}[n + i, [3 \ell + 3j - 2, 3 \ell +  3j]] = 111,
  \end{align*}
  for each $i \in [n]$ and $j \in [\ell]$ (see \Cref{fig:ov} for an illustration).
  We show that $\maxd(\mathbf{T}) = 5 \ell$ if and only if there are $i, i' \in [n]$ such that $u_i$ and~$v_{i'}$ are orthogonal.
  By construction, we have
  \begin{align*}
    \dist(\mathbf{T}[i, [3j - 2, 3j]], \mathbf{T}[n + i', [3j - 2, 3j]]) =
    \begin{cases}
      2 & \text{if } u_i[j] = 0 \text{ or } v_i[j] = 0, \\
      0 & \text{otherwise}.
    \end{cases}
  \end{align*}
  for any $i, i' \in [n]$ and $j \in [\ell]$.
  Thus, it holds for any orthogonal vectors $u_i$ and $v_i'$ that $\dist(\mathbf{T}[i], \mathbf{T}[n + i']) = 5 \ell$.
  Conversely, suppose that there exist $i < i' \in [2n]$ such that $\dist(\mathbf{T}[i], \mathbf{T}[i']) = 5 \ell$.
  It is easy to see that $i \in [n]$ and $i' \in [n + 1, 2n]$ hold (since otherwise $\dist(\mathbf{T}[i], \mathbf{T}[i']) \le 3 \ell$).
  Then, the vectors $u_i$ and $v_{i' - n}$ are orthogonal.
\end{proof}

\begin{figure}
  $$
    \mathbf{T} = \left[
    \begin{array}{c|c|c|c}
      001 & 111 & 001 & 000000000 \\
      111 & 111 & 001 & 000000000 \\
      \hline
      111 & 111 & 010 & 111111111 \\
      111 & 010 & 111 & 111111111 \\
    \end{array}
    \right]
  $$
  \caption{An illustration of the reduction from \textsc{Orthogonal Vectors}, where $\mathcal{U} = \{ 010, 110\}$ and $\mathcal{V} = \{ 110, 101 \}$.}
  \label{fig:ov}
\end{figure}

\subsection{NP-hardness}
\label{sec:unbounded:nph}

Next, we prove the following two NP-hardness results.
In particular, it turns out that \mc{} is NP-hard even for $\alpha = \beta$ when $\alpha$ and $\beta$ are unbounded.
This is in contrast to \Cref{theorem:aapoly}, where we showed that \mc{} is polynomial-time solvable when $\alpha = \beta$ is fixed.

\begin{restatable}{theorem}{knphard}
  \label{theorem:k2nphard}
  \mc{} is NP-hard
  \begin{enumerate}[label=(\roman*)]
    \item \label{ktwo} for $k = 2$ and $\beta - \alpha \ge 3$, and
    \item \label{kthree} for $k = 3$ and $\alpha = \beta$.
  \end{enumerate}

\end{restatable}

The proof for \Cref{theorem:k2nphard} is based on reductions from NP-hard variants of \textsc{3-SAT}.
The most challenging technical aspect of the reductions is to ensure the upper and lower bounds on the pairwise row Hamming distances.
To overcome this challenge, we adjust pairwise row distances by making heavy use of a specific matrix, in which one pair of rows has distance exactly two greater than any other:

\begin{restatable}{lemma}{matrixb}
  \label{lemma:matrixb}
  For each $n \ge 3$ and $i < i' \in [n]$, one can construct in $n^{O(1)}$ time, a matrix $\mathbf{B}^n_{i, i'}\in\{0,1\}^{n\times\ell}$ with $n$ rows and $\ell:=(\binom{n}{2} - 1) (2n - 1)$ columns such that for all $h \ne h' \in [n]$,
  \begin{align*}
    \dist(\mathbf{B}^n_{i, i'}[h], \mathbf{B}^n_{i, i'}[h']) = \begin{cases}
      \gamma(\mathbf{B}^n_{i, i'}) + 2 & \text{if } (h, h') = (i, i'), \\
      \gamma(\mathbf{B}^n_{i, i'})     & \text{otherwise}. \\
    \end{cases}
  \end{align*}
\end{restatable}

\begin{proof}
  First, we define a binary matrix $\mathbf{A}^n \in \{ 0, 1 \}^{n \times (2n - 1)}$ as follows:
  \[
    \mathbf{A}^n := \left[\begin{array}{*{11}{c}}
      1 & 0 & 0                  & 0 & 0 & \cdots & 0 & 0 & 0 & \cdots & 0 \\
      0 & 1 & 0                  & 0 & 0 & \cdots & 0 & 0 & 0 & \cdots & 0 \\
      1 & 1 & 1                  & \multicolumn{4}{c}{\multirow{4}{*}{$\mathbf{I}$}} & \multicolumn{4}{c}{\multirow{4}{*}{$\mathbf{I}$}} \\
      1 & 1 & 1                  &   &   &        &   &   &   &        &   \\
      \multicolumn{3}{c}{$\vdots$} & &   &        &   &   &   &        & \\
      1 & 1 & 1                  &   &   &        &   &   &   &        &   \\
    \end{array}\right],
  \]
  where $\mathbf{I}$ is the \( (n - 2) \times (n - 2) \) identity matrix.
  Note that $\dist(\mathbf{A}^n[1], \mathbf{A}^n[2]) = 2$ and \( \dist(\mathbf{A}^n[h], \mathbf{A}^n[h']) = 4 \) for all \( h < h' \in [n] \) with $(h, h') \ne (1, 2)$.
  We also define the matrix $\mathbf{A}^n_{h, h'}$ obtained from \( \mathbf{A}^n \) by swapping the row vectors \( \mathbf{A}^n[1] \) (and \( \mathbf{A}^n[2] \)) with \( \mathbf{A}^n[h] \) (and \( \mathbf{A}^n[h'] \), respectively) for each $h < h' \in [n]$.
  The matrix $\mathbf{A}^n_{h, h'}$ is a matrix in which the distance between the $h$-th and $h'$-th row vectors are exactly two smaller than all other pairs.
  Now we use the matrix $\mathbf{A}^n_{h, h'}$ to obtain a binary matrix in which the distance of a certain pair of row vectors is exactly two greater than all others.
  We define $\mathbf{B}^n_{i, i'} \in \{ 0, 1 \}^{n \times \ell}$ as the matrix obtained by horizontally stacking \( \binom{n}{2} - 1 \) matrices $\mathbf{A}^n_{h, h'}$ for all $h < h' \in [n]$ with $(h, h') \ne (i, i')$:
  \begin{align*}
    \mathbf{B}^n_{i, i'} := 
    \begingroup
    \setlength\arraycolsep{2pt}
    \begin{bmatrix}
      \mathbf{A}^n_{1, 1} \cdots  \mathbf{A}^n_{1, n}  \cdots 
      \mathbf{A}^n_{i, i + 1}  \cdots  \mathbf{A}^n_{i, i' - 1}  \mathbf{A}^n_{i, i' + 1}  \cdots  \mathbf{A}^n_{i, n} \cdots
       \mathbf{A}^n_{n - 1, n}
    \end{bmatrix}.
    \endgroup
  \end{align*}
  Observe that $\dist(\mathbf{B}^n_{i, i'}[i], \mathbf{B}^n_{i, i'}[i']) = 4 \cdot (\binom{n}{2} - 1) = 2n(n - 1) - 4$, since \( \dist(\mathbf{A}^n_{i, i'}[h], \mathbf{A}^n_{i, i'}[h']) = 4 \) for all $h < h' \in [n]$ with $(h, h') \ne (i, i')$.
  Note also that for each $h < h' \in [n]$ with $(h, h') \ne (i, i')$, we have $\dist(\mathbf{B}^n_{i, i'}[h], \mathbf{B}^n_{i, i'}[h']) = 2n(n - 1) - 6$ because the distance between \( \mathbf{A}^n_{i, i'}[\widetilde{h}] \) and \( \mathbf{A}^n_{i, i'}[\widetilde{h'}] \) is four for every $\widetilde{h} < \widetilde{h'} \in [n]$ except that it is smaller by two for the pair \( \mathbf{A}^n_{i, i'}[i] \) and \( \mathbf{A}^n_{i, i'}[i'] \).
  It is easy to see that the matrix $\mathbf{B}^n_{i, i'}$ can be constructed in polynomial time.
\end{proof}


\begin{proof}[\Cref{theorem:k2nphard}]
  \ref{ktwo}
  We reduce from the following NP-hard variant of \textsc{3-SAT} known as \textsc{(3,\,B2)-SAT}~\cite{BKS03}:

  \dprob
  {(3,\,B2)-SAT}
  {A Boolean formula in conjunctive normal form, in which each literal occurs exactly twice and each clause contains exactly three literals of distinct variables.}
  {Is there a satisfying truth assignment?}

  We divide our proof into two parts as follows.
  We first provide a set $\mathcal{C}$ of incomplete matrices and describe certain completion rules such that
  the given formula of \textsc{(3,\,B2)-SAT} is satisfiable if and only if the matrices \( \mathcal{C} \) can be completed under those rules.
  We then show that one can construct in polynomial time a single incomplete matrix $\mathbf{S}$ containing each matrix in \( \mathcal{C} \) as a submatrix, such that $\mathbf{S}$ admits a solution
  if and only if the submatrices in \( \mathcal{C} \) can be completed according to the rules.
  We are going to exploit the matrix \( \mathbf{B}^n_{i, i'} \) of \Cref{lemma:matrixb} for this construction.

  \paragraph*{Part I}
  Let $\phi$ be an instance of \textsc{(3,\,B2)-SAT} with clauses $C_0, \dots, C_{m - 1}$.
  We define the following matrix for each clause~$C_i$
  \begin{align*}
    \mathbf{C}_i := \left[\begin{array}{*{8}{c}}
      \multicolumn{2}{c}{l^1_i} & 0 & 0                     & 0 & 0 & 1 & 1 \\
      0 & 0                     & \multicolumn{2}{c}{l^2_i} & 0 & 0 & 1 & 0 \\
      0 & 0                     & 0 & 0                     & \multicolumn{2}{c}{l^3_i} & 0 & 1 \\
      1 & 0                     & 1 & 0                     & 1 & 0                     & \multicolumn{2}{c}{c_i}\\
    \end{array}\right].
  \end{align*}
  Here we use $l^1_i, l^2_i, l^3_i$, and $c_i$ to represent two missing entries for notational purposes.
  Note that the matrices $\mathbf{C}_i$ are identical for all $i \in [0, m - 1]$.
  We will prove that $\phi$ is satisfiable if and only if it is possible to complete matrices $\mathcal{C} := \{ \mathbf{C}_i \mid i \in [0, m - 1] \}$ while satisfying the following constraints:
  \begin{enumerate}
    \item The missing entries \( l_i^j \) are filled by 10 or 01 for each $i \in [0, m - 1]$ and $j \in [3]$. \label{item:rule:1}
    \item The missing entries \( c_i \) are filled by 00, 01, or 10 for each $i \in [0, m - 1]$. \label{item:rule:2}
    \item If the missing entries \( c_i \) are filled by 00 (01, 10), then $l_i^1$ ($l_i^2$, $l_i^3$, respectively) are filled by 10 for each $i \in [0, m - 1]$. \label{item:rule:3}
    \item
    Let $\mathcal{Z}$ be the set such that $(i, j, i', j') \in \mathcal{Z}$ if and only if the $j$-th literal in $C_i$ and the $j'$-th literal in $C_{i'}$ correspond to the same variable and one is the negation of the other for each $i < i' \in [0, m - 1]$ and $j, j' \in [3]$.
    If $(i, j, i', j') \in \mathcal{Z}$, then either $l_i^j$ or $l_{i'}^{j'}$ is filled by 01. \label{item:rule:4}
  \end{enumerate}

  Note that there are three choices for filling in $c_i$ by Constraint~\ref{item:rule:2}.
  The intuitive idea is that the choice of $c_i$ dictates which literal (in binary encoding) in the clause~$C_i$ is satisfied.
  We then obtain a satisfying truth assignment for $\phi$, as we shall see in the following claim.
  \begin{claim}
    \label{claim:part1}
    The formula $\phi$ is satisfiable if and only if the matrices $\mathcal{C}$ can be completed according to Constraints~\ref{item:rule:1}~to~\ref{item:rule:4}.
  \end{claim}
  \begin{claimproof}
  \( (\Rightarrow) \)
  If there exists a truth assignment $\tau$ satisfying $\phi$, then at least one literal in the clause $C_i$ evaluates to true for each $i \in [0, m - 1]$.
  We choose an arbitrary number $l_i \in [3]$ such that the $l_i$-th literal of \( C_i \) is satisfied in $\tau$ for each $i \in [0, m - 1]$.
  For each \( i \in [0, m - 1] \) we complete the matrix $C_i$ as follows:
  \begin{itemize}
    \item
    If $l_i = 1$, then the missing entries $c_i, l_i^1, l_i^2, l_i^3$ are filled by 00, 10, 01, 01, respectively.
    \item
    If $l_i = 2$, then the missing entries $c_i, l_i^1, l_i^2, l_i^3$ are filled by 01, 01, 10, 01, respectively.
    \item
    If $l_i = 3$, then the missing entries $c_i, l_i^1, l_i^2, l_i^3$ are filled by 10, 01, 01, 10, respectively.
  \end{itemize}
  It is easy to verify that Constraints~\ref{item:rule:1}~to~\ref{item:rule:3} are satisfied.
  We claim that Constraint~\ref{item:rule:4} is also satisfied.
  Suppose to the contrary that there exists an $(i, j, i', j') \in \mathcal Z$ such that the missing entries $l_i^j$ and~$l_{i'}^{j'}$ are both filled by 10.
  Then, we have $l_i = j$ and $l_{i'} = {j'}$, meaning that~$\tau$ satisfies both $x$ and~$\neg x$ (a contradiction).

  \( (\Leftarrow) \)
  For each \( i \in [0, m - 1] \) and \( j \in [3] \) where $l_i^j$ is filled by 10, we construct a truth assignment such that the $j$-th literal of $C_i$ is satisfied.
  No variable is given opposing truth values by such a truth assignment because of Constraint~\ref{item:rule:4}.
  It also satisfies every clause:
  Otherwise, there exists an integer \( i \in [0, m - 1] \) such that all $l_i^1, l_i^2, l_i^3$ are completed by~01 due to Constraint~\ref{item:rule:1}.
  Now we have a contradiction because Constraints~\ref{item:rule:2}~and~\ref{item:rule:3} imply that at least one of $l_i^1, l_i^2, l_i^3$ is filled by 10.
  \end{claimproof}

  \paragraph*{Part II}
  We provided matrices \( \mathcal{C} \) as well as the constraints on the completion of $\mathcal{C}$ in Part~I.
  Now, we describe how to construct a matrix \( \mathbf{S} \) that admits a completion \( \mathbf{T} \) with $\mind(\mathbf{T}) \ge \alpha$ and $\maxd(\mathbf{T}) \le \beta$ if and only if \( \mathcal{C} \) can be completed fulfilling Constraints~\ref{item:rule:1}~to~\ref{item:rule:4}.
  First, for each matrix $\mathbf{C}_i \in \mathcal{C}$, we introduce a matrix $\mathbf{C}_i' \in \{ 0, 1, \square \}^{11 \times 8}$ containing $\mathbf{C}_i$ by adding row vectors as follows (see \Cref{fig:matrixc1}).
  These additional rows will help to encode Constraints~\ref{item:rule:1}, \ref{item:rule:2}, and \ref{item:rule:3}. 
  \begin{itemize}
    \item
      The first four row vectors of $\mathbf{C}_i'$ are identical to the row vectors of $\mathbf{C}_i$.
    \item
      The row vectors $\mathbf{C}_i'[5]$, $\mathbf{C}_i'[6]$, and $\mathbf{C}_i'[7]$ are obtained by completing the missing entries in $\mathbf{C}_i[1]$, $\mathbf{C}_i[2]$, and $\mathbf{C}_i[3]$, respectively, with 00.
    \item
      The row vectors $\mathbf{C}_i'[8]$, $\mathbf{C}_i'[9]$, and $\mathbf{C}_i'[10]$ are obtained by completing the missing entries in $\mathbf{C}_i[1]$, $\mathbf{C}_i[2]$, and $\mathbf{C}_i[3]$, respectively, with 11.
    \item 
      The row vector $\mathbf{C}_i'[11]$ is obtained by completing the missing entries in $\mathbf{C}_i[4]$ with 00.
  \end{itemize}
  \begin{figure}[t]
    \centering
    \begin{subfigure}{.4\textwidth}
      \begin{align*}
          \left[\begin{array}{*{8}{c}}
          \multicolumn{2}{c}{l^1_i} & 0 & 0                     & 0 & 0 & 1 & 1 \\
          0 & 0                     & \multicolumn{2}{c}{l^2_i} & 0 & 0 & 1 & 0 \\
          0 & 0                     & 0 & 0                     & \multicolumn{2}{c}{l^3_i} & 0 & 1 \\
          1 & 0                     & 1 & 0                     & 1 & 0                     & \multicolumn{2}{c}{c_i}\\
          0 & 0 & 0 & 0 & 0 & 0 & 1 & 1 \\
          0 & 0 & 0 & 0 & 0 & 0 & 1 & 0 \\
          0 & 0 & 0 & 0 & 0 & 0 & 0 & 1 \\
          1 & 1 & 0 & 0 & 0 & 0 & 1 & 1 \\
          0 & 0 & 1 & 1 & 0 & 0 & 1 & 0 \\
          0 & 0 & 0 & 0 & 1 & 1 & 0 & 1 \\
          1 & 0 & 1 & 0 & 1 & 0 & 0 & 0 \\
        \end{array}\right]
      \end{align*}
    \end{subfigure}
    \quad
    \begin{subfigure}{.55\textwidth}
      \begin{align*}
        \begin{array}{*{9}{c}}
                 & \vdots &        &    &        &    &        & \vdots      & \\
          \cdots & l_i^j  & \cdots & 00 & \cdots & 00 & \cdots & 01          & \cdots \\
                 & \vdots &        &    &        &    &        & \vdots      & \\
                 & 00     &        &    &        &    &        & 00          & \\
                 & \vdots &        &    & \ddots &    &        & \vdots      & \\
                 & 00     &        &    &        &    &        & 00          & \\
                 & \vdots &        &    &        &    &        & \vdots      & \\
          \cdots & 01     & \cdots & 00 & \cdots & 00 & \cdots & l_{i'}^{j'} & \cdots \\
                 & \vdots &        &    &        &    &        & \vdots      &
        \end{array}
      \end{align*}
    \end{subfigure}
    \caption{The matrix $\mathbf{C}_i' \in \{ 0, 1, \square \}^{11 \times 8}$ (left).
    The rows $\{ 11i + j, 11i' + j' \}$ and the columns $\{ 8i + 2j - 1, 8i + 2j, 8j' + 2j' - 1, 8j' + 2j' \}$ of $\mathbf{C}$ for $(i, j, i', j') \in \mathcal Z$ (right).}
    \label{fig:matrixc1}
  \end{figure}

  Next, we construct a matrix $\mathbf{C} \in \{ 0, 1, \square \}^{11m \times 8m}$ from the matrices $\mathbf{C}_i'$ as follows (see also \Cref{fig:matrixc1}):
  We start with an empty matrix of size $11m \times 8m$.
  We first place $\mathbf{C}_0', \dots, \mathbf{C}_{m - 1}'$ on the diagonal.
  Then, we place 01 at the intersection of the row containing $l_i^j$ ($l_{i'}^{j'}$) and the columns containing $l_{i'}^{j'}$ ($l_i^j$, respectively), for each $(i, j, i', j') \in \mathcal Z$.
  This essentially encodes Constraint~\ref{item:rule:4}.
  Finally, let the remaining entries be all 0.
  The formal definition is given as follows:
  \begin{itemize}
    \item
      $\mathbf{C}[[11i + 1, 11i + 11], [8i + 1, 8i + 8]] = \mathbf{C}_i'$ for each $i \in [0, m - 1]$.
    \item
      $\mathbf{C}[11i + j, 8i' + 2j'] = 1$ and $\mathbf{C}[11i' + j', 8i + 2j] = 1$ for each $(i, j, i', j') \in \mathcal Z$.
    \item
      All other entries are 0.
  \end{itemize}
  Let $n = 11m$ be the number of rows in $\mathbf{C}$.
  Now we define seven ``types'' $H_1, \dots, H_7$ of row index pairs.
  The first four types correspond to Constraints~\ref{item:rule:1} to \ref{item:rule:4}.
  In \Cref{claim:matrixs}, we show how to enforce Constraints~\ref{item:rule:1} to \ref{item:rule:4} by appending an appropriate number of matrices given in \Cref{lemma:matrixb}.
  The other three types are defined based on the number of missing entries.
  For each $h < h' \in [n]$,
  \begin{itemize}
    \item
      $(h, h') \in H_1$ if $h = 11i + j$ and $h' = 11i + j'$ for some $i \in [0, m - 1]$ and $(j, j') \in \{ (1, 5), (2, 6), (3, 7), \allowbreak (1, 8), (2, 9), (3, 10) \}$.
    \item
      $(h, h') \in H_2$ if $h = 11i + j$ and $h' = 11i + j'$ for some $i \in [0, m - 1]$ and $(j, j') = (4, 11)$.
    \item
      $(h, h') \in H_3$ if $h = 11i + j$ and $h' = 11i + j'$ for some $i \in [0, m - 1]$ and $(j, j') \in \{ (1, 4), (2, 4), (3, 4) \}$.
    \item
      $(h, h') \in H_4$ if $h = 11i + j$ and $h' = 11i' + j'$ for $(i, j, i', j') \in \mathcal Z$.
    \item
      For each $h < h' \in [n]$ with $(h, h') \not\in H_1, \dots, H_4$,
      \begin{itemize}
        \item
          $(h, h') \in H_5$ if both $\mathbf{C}[h]$ and $\mathbf{C}[h']$ have missing entries.
        \item
          $(h, h') \in H_6$ if exactly one of $\mathbf{C}[h]$ and $\mathbf{C}[h']$ have missing entries.
        \item
          $(h, h') \in H_7$ if neither of $\mathbf{C}[h]$ and $\mathbf{C}[h']$ has missing entries.
      \end{itemize}
  \end{itemize}

  For each type of row index pairs, we ``adjust'' the pairwise distances using \Cref{lemma:matrixb} to encode Constraints~\ref{item:rule:1} to \ref{item:rule:4}.
  Before doing so, we prove an auxiliary claim stating that $\maxd(\mathbf{C})\le 8$.
  \begin{claim}
    \label{claim:cdiameter}
    $\dist(\mathbf{C}[h], \mathbf{C}[h']) \le 8$ for each $h < h' \in [n]$.
  \end{claim}
  \begin{claimproof}
    It suffices to show that $\dist(\mathbf{C}[h], 0^{8m}) \le 4$ for each $h \in [n]$, because we then have $\dist(\mathbf{C}[h], \mathbf{C}[h']) \le \dist(\mathbf{C}[h], 0^{8m}) + \dist(\mathbf{C}[h'], 0^{8m}) \le 8$ by the triangle inequality.
    Suppose that $h = 11i + j$ for $i \in [0, m - 1]$ and $j \in [3]$.
    Then, the row vector $\mathbf{C}[h]$ contains at most two~1's in~$\mathbf{C}_i'$, and exactly two~1's elsewhere because each literal appears in the formula $\phi$ exactly twice.
    It follows that $\mathbf{C}[h]$ contains at most four 1's.
    Hence, we can assume that $h = 11i + j$ for $i \in [0, m - 1]$ and $j \in [4, 11]$.
    Then, the row vector $\mathbf{C}[h]$ contains at most four~1's, because all~1's appear in $\mathbf{C}_i'$.
    This shows the claim.
  \end{claimproof}

  \begin{claim}
    \label{claim:matrixs}
    There exists a $\beta \in \NN$ and a complete matrix $\mathbf{D}$ over $\{ 0, 1 \}$ with $n$ rows such that all of the following hold for $\mathbf{S} = \begin{bmatrix} \mathbf{C} \quad \mathbf{D} \end{bmatrix}$:
    \begin{itemize}
      \item
        \( \dist(\mathbf{S}[h], \mathbf{S}[h']) = \beta - 1 \) for each $(h, h') \in H_1$ (cf.\ Constraint~\ref{item:rule:1}).
      \item
        \( \dist(\mathbf{S}[h], \mathbf{S}[h']) = \beta - 1 \) for each $(h, h') \in H_2$ (cf.\ Constraint~\ref{item:rule:2}).
      \item
        \( \dist(\mathbf{S}[h], \mathbf{S}[h']) = \beta - 3 \) for each $(h, h') \in H_3$ (cf.\ Constraint~\ref{item:rule:3}).
      \item
        \( \dist(\mathbf{S}[h], \mathbf{S}[h']) \in \{ \beta - 3, \beta - 2 \} \) for each $(h, h') \in H_4$ (cf.\ Constraint~\ref{item:rule:4}).
      \item
        \( \dist(\mathbf{S}[h], \mathbf{S}[h']) \in \{ \beta - 3, \beta - 2 \} \) for each $(h, h') \in H_5$.
      \item
        \( \dist(\mathbf{S}[h], \mathbf{S}[h']) \in \{ \beta - 2, \beta - 1 \} \) for each $(h, h') \in H_6$.
      \item
        \( \dist(\mathbf{S}[h], \mathbf{S}[h']) \in \{ \beta - 1, \beta \} \) for each $(h, h') \in H_7$.
    \end{itemize}
  \end{claim}
  \begin{claimproof}
    We obtain the matrix $\mathbf{D}$ by horizontally stacking \( \mathbf{B}_{h, h'}^{n} \) of \Cref{lemma:matrixb} $c_{h, h'}$ times (where $c_{h, h'} \in \NN$ is to be defined)  for each $h < h' \in [n]$.
    Recall that \( \dist(\mathbf{B}^n_{h, h'}[i], \mathbf{B}^n_{h, h'}[i']) \) equals $\mind(\mathbf{B}^n_{h, h'}) + 2$ if $(h, h') = (i, i')$ and $\mind(\mathbf{B}^n_{h, h'})$ otherwise and let
    \begin{align*}
      \beta = \sum_{h < h' \in [n']} c_{h, h'} \cdot \mind(\mathbf{B}^n_{h, h'}) + 11.
    \end{align*}
    Observe that the pairwise row distance can be rewritten as follows for each $h < h' \in [n]$:
    \begin{align*}
      \dist(\mathbf{S}[h], \mathbf{S}[h'])
      &= \dist(\mathbf{C}[h], \mathbf{C}[h']) + c_{h, h'} \cdot (\mind(\mathbf{B}^n_{h, h'}) + 2) + \sum_{\substack{i < i' \in [n], \\ (i, i') \ne (h, h')}} c_{i, i'} \cdot \mind(\mathbf{B}^n_{h, h'})  \\
      &= \dist(\mathbf{C}[h], \mathbf{C}[h']) + 2 c_{h, h'} + \beta - 11.
    \end{align*}
    We define $c_{h, h'}$ for each $(h, h') \in H_1 \cup H_2 \cup H_3$ as follows.
    \begin{itemize}
      \item
        Let $c_{h, h'} = 4$ for each $(h, h') \in H_1$.
        Then, we have $\dist(\mathbf{S}[h], \mathbf{S}[h']) = 2 + 2 \cdot 4 - \beta - 11 = \beta - 1$.
      \item
        Let $c_{h, h'} = 5$ for each $(h, h') \in H_2$.
        Then, we have $\dist(\mathbf{S}[h], \mathbf{S}[h']) = 0 + 2 \cdot 5 - \beta - 11 = \beta - 1$.
      \item
        Let $c_{h, h'} = 2$ for each $(h, h') \in H_3$.
        Then, we have $\dist(\mathbf{S}[h], \mathbf{S}[h']) = 4 + 2 \cdot 2 - \beta - 11 = \beta - 3$.
    \end{itemize}

    For the remainder (that is, $(h, h') \in H_4 \cup \dots \cup H_7$), it has to be shown that there exists $c_{h, h'} \in \NN$ such that $\dist(\mathbf{S}[h], \mathbf{S}[h']) \in \{ x, x + 1 \}$ holds for $x \in \NN$ with $x \ge \beta - 3$.
    Let $c_{h, h'} = \lceil (11 + x - \beta - \dist(\mathbf{C}[h], \mathbf{C}[h'])) / 2 \rceil$.
    Clearly, $c_{h, h'}$ is an integer and it holds that $c_{h, h'} \ge 0$ because $x - \beta \ge -3$ and $\dist(\mathbf{C}[h], \mathbf{C}[h']) \le 8$ by \Cref{claim:cdiameter}.
    Moreover, we have $\dist(\mathbf{S}[h], \mathbf{S}[h']) = x$ if $11 + x - \beta - \dist(\mathbf{C}[h], \mathbf{C}[h'])$ is even and $\dist(\mathbf{S}[h], \mathbf{S}[h']) = x + 1$ otherwise.
  \end{claimproof}
  
  Finally, we show that Constraints~\labelcref{item:rule:1,item:rule:2,item:rule:3,item:rule:4} are essentially the same as the pairwise row distance constraints on the matrix $\mathbf{S}$ of \Cref{claim:matrixs}.
  \begin{claim}
    \label{claim:part2}
    The matrices $\mathcal{C}$ can be completed according to Constraints~\labelcref{item:rule:1,item:rule:2,item:rule:3,item:rule:4}, if and only if~$\mathbf{S}$ admits a completion $\mathbf{T}$ with $\beta -3 \le \maxd(\mathbf{T}) \le \beta$.
  \end{claim}
  \begin{claimproof}
    $(\Rightarrow)$
    Let $\mathbf{T}$ be the matrix where the missing entries of $\mathbf{S}$ are filled as in the completion of~$\mathcal{C}$.
    First, note that $\mind(\mathbf{T}) \ge \mind(\mathbf{S}) \ge \beta - 3$ by \Cref{claim:matrixs}.
    We show that $\dist(\mathbf{T}[h], \mathbf{T}[h']) \le \beta$ holds for each $h < h' \in [n]$.
    \begin{itemize}
      \item
        Suppose that $(h, h') \in H_1$.
        Then, the missing entries in $\mathbf{S}[h]$ are filled by 10 or 01 by Constraint~\ref{item:rule:1} and $\mathbf{S}[h']$ (which has no missing entries) has 00 or 11 in the corresponding positions.
        Hence, $\dist(\mathbf{T}[h], \mathbf{T}[h']) = \dist(\mathbf{S}[h], \mathbf{S}[h']) + 1 = \beta$.
      \item
        Suppose that $(h, h') \in H_2$.
        Then, the missing entries in $\mathbf{S}[h]$ are filled by 00, 01, or 10 by Constraint~\ref{item:rule:2} and $\mathbf{S}[h']$ (which has no missing entries) has 00 in the corresponding positions.
        Hence, $\dist(\mathbf{T}[h], \mathbf{T}[h']) = \dist(\mathbf{S}[h], \mathbf{S}[h']) + 1 = \beta$.
      \item
        Suppose that $(h, h') \in H_3$.
        Note that $\mathbf{S}[h]$ has missing entries $l_i^j$ and $\mathbf{S}[h']$ has missing entries $c_i$ for $i \in [0, m - 1]$ and $j \in [3]$.
        Let $c_i^1 c_i^2$ be the completion of $c_i$ for $c_i^1, c_i^2 \in \{ 0, 1 \}$.
        If $\mathbf{T}[h]$ has $1 - c_i^1$ and $1 - c_i^2$ in the corresponding positions, then $\dist(\mathbf{T}[h], \mathbf{T}[h']) = \dist(\mathbf{S}[h], \mathbf{S}[h']) + 2 = \beta - 1$.
        Otherwise, $\mathbf{T}[h]$ matches in at least one position where $\mathbf{T}[h']$ has missing entries $c_i$.
        Therefore, $\dist(\mathbf{T}[h], \mathbf{T}[h']) \le \dist(\mathbf{S}[h], \mathbf{S}[h']) + 3 = \beta$.
      \item
        Suppose that $(h, h') \in H_4$.
        Note that $\mathbf{S}[h]$ has missing entries $l_i^j$ and $\mathbf{S}[h']$ has missing entries $l_{i'}^{j'}$.
        Also note that $\mathbf{S}[h]$ and $\mathbf{S}[h']$ have 01 where the other row vector has missing entries.
        Since either $l_i^j$ or $l_{i'}^{j'}$ must be completed by 01 due to Constraint~\labelcref{item:rule:4}, we have $\dist(\mathbf{T}[h], \mathbf{T}[h']) = \dist(\mathbf{S}[h], \mathbf{S}[h']) \le \beta - 2 + 2 = \beta$.
      \item
        Suppose that $(h, h') \in H_5 \cup H_6 \cup H_7$.
        Let $x \in \{ 0, 1, 2 \}$ be the number of row vectors with missing entries in $\{ \mathbf{S}[h], \mathbf{S}[h'] \}$.
        Then, we have $\dist(\mathbf{S}[h], \mathbf{S}[h']) \in \{ \beta - x - 1, \beta - x \}$.
        If $\mathbf{S}[h]$ has missing entries, then $\mathbf{S}[h']$ has 00 in the corresponding positions, and vice versa.
        Since the missing entries are filled by 00, 01, or 10 according to Constraints~\labelcref{item:rule:1,item:rule:2,item:rule:3}, we have $\dist(\mathbf{T}[h], \mathbf{T}[h']) \le \dist(\mathbf{S}[h], \mathbf{S}[h']) + x = \beta$.
    \end{itemize}

    $(\Leftarrow)$
    We complete the matrices in $\mathcal{C}$ in the same way as in the completion of $\mathbf{S}$.
    It is easy to verify all Constraints~\labelcref{item:rule:1,item:rule:2,item:rule:3,item:rule:4} are satisfied.
  \end{claimproof}
  Note that $\mind(\mathbf{T}) \ge \mind(\mathbf{S}) \ge \beta - 3$ for any completion $\mathbf{T}$ of $\mathbf{S}$.
  Hence, it follows from Claims~\ref{claim:part1} and~\ref{claim:part2} that $\phi$ is satisfiable if and only if the \textsc{DMC} instance $(\mathbf{S},\alpha,\beta)$ is a \textbf{Yes}-instance, for any $\alpha \le \beta - 3$.
  This concludes the proof of \Cref{theorem:k2nphard}~\ref{ktwo}.

  \ref{kthree} To prove that \mc{} is NP-hard for $\alpha = \beta$ and $k = 3$, we provide a polynomial-time reduction from another NP-hard variant of \textsc{3-SAT}~\cite{MR01}:

  \dprob
  {Cubic Monotone 1-in-3 SAT}
  {A Boolean formula in conjunctive normal form, in which each variable appears exactly three times and each clause contains exactly three distinct positive literals.}
  {Is there a truth assignment that satisfies exactly one literal in each clause?}

  Our reduction heavily depends on the fact that $\alpha = \beta$.
  This is contrary to the reduction in part~\ref{ktwo}, which in fact works for any $\alpha \le \beta - 3$.

  Let $\phi$ be an instance of \textsc{Cubic Monotone 1-in-3 SAT}.
  Our proof has two parts:
  First, we provide an incomplete matrix $\mathbf{C}$ and we show that $\phi$ is a \textbf{Yes}-instance if and only if $\mathbf{C}$ can be completed under certain constraints.
  Then, we obtain an instance $(\mathbf{S}, \alpha, \alpha)$ of \mc{} by adjusting the pairwise row distances with the help of \Cref{lemma:matrixb}.

  Suppose that $\phi$ contains variables $x_1, \dots, x_m$ and clauses $C_1, \dots, C_m$, where $C_i = (C_i^1 \vee C_i^2 \vee C_i^3)$ for each $i \in [m]$.
  First, we define matrices $\mathbf{C}_1, \mathbf{C}_3 \in \{ 0, 1, \square \}^{m \times 2m}$ and $\mathbf{C}_2, \mathbf{C}_4 \in \{ 0, 1, \square \}^{m \times 3m}$.
  The incomplete matrices $\mathbf{C}_1$ and $\mathbf{C}_4$ and represent variables and clauses, respectively.
  The matrices $\mathbf{C}_2$ and $\mathbf{C}_3$ are complete.
  We use $a_i$ (and $b_i$) to represent two (three, respectively) missing entries in $\mathbf{C}_1$ ($\mathbf{C}_4$, respectively) for each $i \in [m]$.
  For each $i \in [m]$ and $j \in [m]$, let
  \begin{align*}
    &\mathbf{C}_1[i, \{ 2j - 1, 2j \}] = \begin{cases}
      a_i & \text{if } i = j, \\
      00 & \text{otherwise}.
    \end{cases}
    &\mathbf{C}_2[i, [3j - 2, 3j]] &= 
    \begin{cases}
      011 & \text{if } x_i = C_j^1, \\
      101 & \text{if } x_i = C_j^2, \\
      110 & \text{if } x_i = C_j^3, \\
      000 & \text{otherwise}.
    \end{cases}
    \\
    &\mathbf{C}_3[i, \{ 2j - 1, 2j \}] = \begin{cases}
      10 & \text{if } x_i \text{ is in } C_j, \\
      00 & \text{otherwise}.
    \end{cases}
    &\mathbf{C}_4[i, [3j - 2, 3j]] &= \begin{cases}
      b_i & \text{if } i = j, \\
      000 & \text{otherwise}.
    \end{cases}
  \end{align*}
  We obtain an incomplete matrix $\mathbf{C} \in \{ 0, 1, \square \}^{(2m + 1) \times (5m + 1)}$ by appending a column vector $(0^m 1^m)^T$ and a row vector $0^{5m + 1}$ to the following matrix
  \[
    \left[\begin{array}{*{2}{c}}
      \mathbf{C}_1 & \mathbf{C}_2 \\ 
      \mathbf{C}_3 & \mathbf{C}_4 \\
    \end{array}\right].
  \]
  Refer to \Cref{fig:matrixc} for an illustration.

  \begin{figure}[t]
    $$
      \mathbf{C} = \left[\begin{array}{*{4}{c}|*{4}{c}|c}
        a_1 & 00  & 00  & 00  \; & 011 & 011 & 011 & 000 \; & 1 \\
        00  & a_2 & 00  & 00  \; & 101 & 101 & 000 & 011 \; & 1 \\
        00  & 00  & a_3 & 00  \; & 110 & 000 & 101 & 101 \; & 1 \\
        00  & 00  & 00  & a_4 \; & 000 & 110 & 110 & 110 \; & 1 \\
        \hline
        10  & 10  & 10  & 00  \; & b_1 & 000 & 000 & 000 \; & 0 \\
        10  & 10  & 00  & 10  \; & 000 & b_2 & 000 & 000 \; & 0 \\
        10  & 00  & 10  & 10  \; & 000 & 000 & b_3 & 000 \; & 0 \\
        00  & 10  & 10  & 10  \; & 000 & 000 & 000 & b_4 \; & 0 \\
        \hline
        00  & 00  & 00  & 00  \; & 000 & 000 & 000 & 000 \; & 0 \\
      \end{array}\right]
    $$
    \caption{An example of $\mathbf{C}$ for $\phi = (x_1 \vee x_2 \vee x_3) \wedge (x_1 \vee x_2 \vee x_4) \wedge (x_1 \vee x_3 \vee x_4) \wedge (x_2 \vee x_3 \vee x_4)$.}
    \label{fig:matrixc}
  \end{figure}

  Intuitively speaking, we will use the first $m$ rows to encode the variables and the following $m$ rows to encode the clauses.
  \begin{claim}\label{claim:1in3S}
    There is a truth assignment that satisfies exactly one literal in clause $C_i$ for each $i \in [m]$ if and only if there is a completion $\mathbf{C'}$ of $\mathbf{C}$ such that
    \begin{enumerate}
      \item
        $\dist(\mathbf{C}'[i], \mathbf{C}[2m + 1]) = \dist(\mathbf{C}[i], \mathbf{C}[2m + 1]) + 1$ for each $i \in [2m]$.
      \item
        $\dist(\mathbf{C}'[i], \mathbf{C}'[m + i']) = \dist(\mathbf{C}[i], \mathbf{C}[m + i']) + 3$ for each $i, i' \in [m]$ such that $x_i$ is in $C_{i'}$.
    \end{enumerate}
  \end{claim}
  \begin{claimproof}
    $(\Rightarrow)$
    Let $\tau$ be a truth assignment satisfying exactly one literal in each clause of $\phi$.
    Consider the matrix $\mathbf{C}'$ obtained by completing $\mathbf{C}$ as follows for each $i \in [m]$:
    \begin{itemize}
      \item
        The missing entries $a_i$ are filled by 10 if $x_i$ is true in $\tau$ and by 01 otherwise.
      \item
        The missing entries $b_i$ are filled by 100 if $C_i^1$ is true in $\tau$.
      \item
        The missing entries $b_i$ are filled by 010 if $C_i^2$ is true in $\tau$.
      \item
        The missing entries $b_i$ are filled by 001 if $C_i^3$ is true in $\tau$.
    \end{itemize}
    It is easy to see that the first constraint of the claim is indeed fulfilled.
    For the other constraint, consider $i, i' \in [m]$ such that $x_i$ is in $C_{i'}$.
    We prove that $\dist(\mathbf{C}'[i], \mathbf{C}'[m + i']) - \dist(\mathbf{C}[i], \mathbf{C}[m + i']) = 3$ or equivalently,
    \begin{align*}
      \dist(a_i', \mathbf{C}[m + i', [2i' - 1, 2i']]) +
      \dist(\mathbf{C}[i, [2m + 3i' - 2, 2m + 3i']], b_{i'}')
      = 3,
    \end{align*}
    where $a_i'$ and $b_{i'}'$ are the completion of $a_i$ and $b_{i'}$ in $\mathbf{C}'$.
    We show that it holds for the case $x_i = C_{i'}^1$.
    It can be proven analogously for the cases of $x_i = C_{i'}^2$ and $x_i = C_{i'}^3$ as well.

    Note that $\mathbf{C}[m + i', [2i' - 1, 2i']] = 10$ and $\mathbf{C}[i, [2m + 3i' - 2, 2m + 3i']] = 011$.
    If $x_i$ is true in $\tau$, then $a_i'$ and $b_{i'}'$ are 10 and 100, respectively.
    Thus, the equality above holds.
    If $x_i$ is false in $\tau$, then $a_i' = 01$ and $b_{i'}' \in \{ 010, 001 \}$.
    Again the equality above holds.

    $(\Leftarrow)$
    Let $a_i'$ and $b_{i}'$ be the completion of $a_i$ and $b_{i}$ in $\mathbf{C}'$ for each $i \in [m]$.
    Due to the first constraint, exactly one entry in $a_i'$ and $b_i'$ must be 1 for each $i \in [m]$.
    Now consider the truth assignment $\tau$ that assigns $x_i$ to true if $a_i' = 10$ and false if $a_i' = 01$ for each $i \in [m]$.
    We show that $\tau$ satisfies exactly one literal in each clause of $\phi$.
    Consider $i \in [m]$ with $C_i = (x_{i_1} \vee x_{i_2} \vee x_{i_3})$.
    By the second constraint, we have
    \begin{align*}
      \sum_{j = 1}^3 \dist(\mathbf{C}'[i_j], \mathbf{C}'[m + i]) - \dist(\mathbf{C}[i_j], \mathbf{C}[m + i]) = 9.
    \end{align*}
    Rewriting the left-hand side in terms of $a_{i_1}', a_{i_2}', a_{i_3}', b_{i}'$, we obtain
    \begin{align*}
       \dist(011, b_{i}') + \dist(101, b_{i}') + \dist(110, b_{i}') + \dist(10, a_{i_1}') + \dist(10, a_{i_2}') + \dist(10, a_{i_3}') = 9.
    \end{align*}
    Since $b_{i}' \in \{ 100, 010, 001 \}$, it follows that the first three terms sum up to exactly $5$ and hence $\dist(10, a_{i_1}') + \dist(10, a_{i_2}') + \dist(10, a_{i_3}') = 4$.
    This means that exactly one of $a_{i_1}', a_{i_2}', a_{i_3}'$ is 10 and the remaining two are 01.
    Thus, exactly one literal in $C_i$ is satisfied.
  \end{claimproof}

  We will build an incomplete matrix $\mathbf{S}$ with $2m + 1$ rows from $\mathbf{C}$ by horizontally appending matrices~$\mathbf{B}_{i, i'}^{2m + 1}$ of \Cref{lemma:matrixb} such that $\mathbf{S}$ admits a completion $\mathbf{T}$ with $\mind(\mathbf{T}) = \maxd(\mathbf{T})$ if and only if $\phi$ has a satisfying assignment.
  To determine how many times we append $\mathbf{B}_{i, i'}^{2m + 1}$, we observe the following about the pairwise distances in $\mathbf{C}$ (see \Cref{fig:matrixc}):
  \begin{itemize}
    \item
      For each $i \in [m]$, $\dist(\mathbf{C}[i], \mathbf{C}[2m + 1]) = 7$ and $\dist(\mathbf{C}[m + i], \mathbf{C}[2m + 1]) = 3$.
    \item
      For each $i \ne i' \in [m]$, $\dist(\mathbf{C}[i], \mathbf{C}[i']) = 12 - 2 c_{i, i'}$ and $\dist(\mathbf{C}[m + i], \mathbf{C}[m + i']) = 6 - 2 c_{i, i'}'$. 
      Here $c_{i, i'} \in \{ 0, 1, 2, 3 \}$ is the number of clauses that contain both $x_i$ and $x_{i'}$, and $c_{i, i'}' \in \{ 0, 1, 2, 3 \}$ is the number of variables that are both in $C_i$ and $C_{i'}$.
    \item
      For each $i, i' \in [m]$ with $x_i$ in $C_{i'}$, $\dist(\mathbf{C}[i], \mathbf{C}[m + i']) = 7$.
    \item
      For each $i, i' \in [m]$ with $x_i$ not in $C_{i'}$, $\dist(\mathbf{C}[i], \mathbf{C}[m + i']) = 10$.
  \end{itemize}

  Now we construct $\mathbf{S}$ as follows.
  Recall that $\dist(\mathbf{B}_{i, i'}^{2m + 1}[i],\mathbf{B}_{i, i'}^{2m + 1}[i']) = \mind(\mathbf{B}_{i, i'}^{2m + 1}) + 2$ and $\dist(\mathbf{B}_{i, i'}^{2m + 1}[h],\mathbf{B}_{i, i'}^{2m + 1}[h']) = \mind(\mathbf{B}_{i, i'}^{2m + 1})$ for all $h < h' \in [2m + 1]$ with $(h, h') \ne (i, i')$.
  For $c \in \NN$, let $c \mathbf{B}_{i, i'}$ be the matrix obtained by horizontally stacking $\mathbf{B}_{i, i'}^{2m + 1}$ $c$ times (we drop the superscript for readability).
  We also compute a value for $\alpha \in \NN$ as follows:
  We start with $\alpha = 14$ and we increase $\alpha$ by $c \cdot \mind(\mathbf{B}_{i, i'}^{2m + 1})$ each time $c \mathbf{B}_{i, i'}$ is appended to~$\mathbf{C}$.
  We horizontally append the following matrices:
  \begin{itemize}
    \item
      $3\mathbf{B}_{i, 2m + 1}$ and $5\mathbf{B}_{m + i, 2m + 1}$ for each $i \in [m]$.
    \item
      $c_{i, i'} \mathbf{B}_{i, i'}$ and  $(c_{i, i'}' + 3) \mathbf{B}_{m + i, m + i'}$ for each $i < i' \in [m]$.
    \item
      $2\mathbf{B}_{i, m + i'}$ for each $i, i' \in [m]$ with $x_i$ in $C_{i'}$.
    \item
      $1\mathbf{B}_{i, m + i'}$ for each $i, i' \in [m]$ with $x_i$ not in $C_{i'}$.
  \end{itemize}

  Note that for each $i, i' \in [2m + 1]$, $\dist(\mathbf{S}[i], \mathbf{S}[i']) = \dist(\mathbf{C}[i], \mathbf{C}[i']) + 2 \cdot n_{i, i'} + \alpha - 14$, where $n_{i, i'}$ is the number of appended $\mathbf{B}_{i, i'}$'s. 
  Thus, the pairwise row distances in $\mathbf{S}$ are given as follows:
  \begin{itemize}
    \item
      For each $i \in [m]$, $\dist(\mathbf{S}[i], \mathbf{S}[2m + 1]) = 7 + 2 \cdot 3 + \alpha - 14 = \alpha - 1$ and $\dist(\mathbf{S}[m + i], \mathbf{S}[2m + 1]) = 3 + 2 \cdot 5 + \alpha - 14 = \alpha - 1$.
    \item
      For each $i \ne i' \in [m]$, $\dist(\mathbf{S}[i], \mathbf{S}[i']) = (12 - 2 c_{i, i'}) + 2\cdot c_{i, i'} + \alpha - 14 = \alpha - 2$ and $\dist(\mathbf{S}[m + i], \mathbf{S}[m + i']) = (6 - 2 c_{i, i'}') + 2 \cdot (c_{i, i'}' + 3) + \alpha - 14 = \alpha - 2$.
    \item 
      For each $i, i' \in [m]$ with $x_i$ in $C_{i'}$, $\dist(\mathbf{S}[i], \mathbf{S}[m + i']) = 7 + 2 \cdot 2 + \alpha - 14 = \alpha - 3$.
    \item 
      For each $i, i' \in [m]$ with $x_i$ not in $C_{i'}$, $\dist(\mathbf{S}[i], \mathbf{S}[m + i']) = 10 + 2 \cdot 1 + \alpha - 14 = \alpha - 2$.
  \end{itemize}

  Finally, we prove that one can complete $\mathbf{C}$ as specified in \Cref{claim:1in3S} if and only if one can complete $\mathbf{S}$ into a matrix~$\mathbf{T}$ with $\mind(\mathbf{T})=\maxd(\mathbf{T})=\alpha$.
  \begin{claim}\label{claim:CS}
    There is a completion $\mathbf{C'}$ of $\mathbf{C}$ such that
    \begin{enumerate}
      \item
        $\dist(\mathbf{C}'[i], \mathbf{C}[2m + 1]) = \dist(\mathbf{C}[i], \mathbf{C}[2m + 1]) + 1$ for each $i \in [2m]$.
      \item
        $\dist(\mathbf{C}'[i], \mathbf{C}'[m + i']) = \dist(\mathbf{C}[i], \mathbf{C}[m + i']) + 3$ for each $i, i' \in [m]$ such that $x_i$ is in $C_{i'}$.
    \end{enumerate}
    if and only if there is a completion $\mathbf{T}$ of $\mathbf{S}$ with $\dist(\mathbf{T}[i], \mathbf{T}[i']) = \alpha$ for all $i \ne i' \in [2m + 1]$.
  \end{claim}
  \begin{claimproof}
    $(\Rightarrow)$
    Consider the completion $\mathbf{T}$ of $\mathbf{S}$ in which each missing entry is completed as in~$\mathbf{C}'$.
    Let $a_{i}' $ and $b_{i}'$ be the completion of $a_i$ and $b_i$ for each $i \in [m]$.
    We have $a_{i}' \in \{ 10, 01 \}$ and $b_{i}' \in \{ 100, 010, 001 \}$ due to the first constraint.
    Now we examine each pairwise row distance.
    \begin{itemize}
      \item
        For each $i \in [m]$, $\dist(\mathbf{T}[i], \mathbf{T}[2m + 1]) = \dist(\mathbf{S}[i], \mathbf{S}[2m + 1]) + \dist(a_i', 00) = (\alpha - 1) + 1 = \alpha$ and $\dist(\mathbf{T}[m + i], \mathbf{T}[2m + 1]) = \dist(\mathbf{S}[m + i], \mathbf{S}[2m + 1]) + \dist(b_i', 00) = (\alpha - 1) + 1 = \alpha$.
      \item
        For each $i \ne i' \in [m]$, $\dist(\mathbf{T}[i], \mathbf{T}[i']) = \dist(\mathbf{S}[i], \mathbf{S}[i']) + \dist(a_i, 00) + \dist(a_{i'}, 00) = (\alpha - 2) + 1 + 1 = \alpha$ and $\dist(\mathbf{T}[m + i], \mathbf{T}[m + i']) = \dist(\mathbf{S}[m + i], \mathbf{S}[m + i']) + \dist(b_i, 000) + \dist(b_{i'}, 000) = (\alpha - 2) + 1 + 1 = \alpha$
      \item 
        For each $i, i' \in [m]$ with $x_i$ in $C_{i'}$, $\dist(\mathbf{T}[i], \mathbf{T}[m + i']) = \dist(\mathbf{S}[i], \mathbf{S}[m + i']) + (\dist(\mathbf{C}'[i], \mathbf{C}'[m + i']) - \dist(\mathbf{C}[i], \mathbf{C}[m + i'])) = (\alpha - 3) + 3 = \alpha$ because of the second constraint.
      \item 
        For each $i, i' \in [m]$ with $x_i$ not in $C_{i'}$, $\dist(\mathbf{T}[i], \mathbf{T}[m + i']) = \dist(\mathbf{S}[i], \mathbf{S}[m + i']) + \dist(a_i, 00) + \dist(b_{i}', 000) = (\alpha - 2) + 1 + 1 = \alpha$.
    \end{itemize}
    Hence, all pairwise row distances are equal to $\alpha$.

    $(\Leftarrow)$
    Consider the completion $\mathbf{C}'$ of $\mathbf{C}$ in which each missing entry is completed as in $\mathbf{T}$.
    For each $i \in [2m]$, we have $\dist(\mathbf{C}'[i], \mathbf{C}[2m + 1]) - \dist(\mathbf{C}[i], \mathbf{C}[2m + 1]) = \dist(\mathbf{T}[i], \mathbf{T}[2m + 1]) - \dist(\mathbf{S}[i], \mathbf{S}[2m + 1]) = \alpha - (\alpha - 1) = 1$.
    Moreover, it holds for each $i, i' \in [m]$ with $x_i$ in $C_{i'}$ that $\dist(\mathbf{C}'[i], \mathbf{C}'[m + i']) - \dist(\mathbf{C}[i], \mathbf{C}[m + i']) = \dist(\mathbf{S}'[i], \mathbf{S}'[m + i']) - \dist(\mathbf{S}[i], \mathbf{S}[m + i']) = \alpha - (\alpha - 3) = 3$.
    This concludes the proof of the claim.
  \end{claimproof}
  Combining \Cref{claim:1in3S} and \Cref{claim:CS}, we have that $\phi$ is a \textbf{Yes}-instance if and only if the \mc{} instance $(\mathbf{S}, \alpha, \alpha)$ is a \textbf{Yes}-instance.
\end{proof}

The problem of deciding whether an incomplete matrix $\mathbf{S} \in \{ 1, -1, \square \}^{n \times n}$ can be completed into a Hadamard matrix~\cite{Joh90}, is equivalent to \textsc{DMC} with $n = \ell$ and $\alpha = \beta = n / 2$.
We conjecture that one can adapt the proof of \Cref{theorem:k2nphard}~\ref{kthree} to show the NP-hardness of this problem.
We also conjecture that \mc{} with~$k=3$ is actually NP-hard for every value of $\beta - \alpha$. Similar reductions might work here as well.
By contrast, we believe the case~$k=2$ and~$\beta - \alpha=1$ to be polynomial-time solvable, again by reducing to \textsc{2-SAT}.

\section{Conclusion}
\label{sec:mindmc:conc}

Together with the recent work of Eiben et al.~\cite{EGKOS21,EGKOS22}, we are 
seemingly among the first in the context of stringology that make extensive use of Deza's theorem and sunflowers.
While Eiben et al.~\cite{EGKOS21,EGKOS22} achieved classification results in terms of parameterized (in)tractability, we conducted a detailed complexity analysis in terms of polynomial-time solvable versus NP-hard cases.
Figure~\ref{fig:results} provides a visual overview on our results for 
\textsc{Diameter Matrix Completion} (DMC), also spotting concrete 
open questions.

Going beyond open questions directly arising from Figure~\ref{fig:results},
we remark that
it is known that the clustering variant of \mc{} can be solved 
in polynomial time when the number of clusters is two and the matrix is complete \cite{GJL04}.
Hence, it is natural to ask whether our tractability results can be extended to this matrix completion clustering problem as well.
Furthermore, we proved that there are polynomial-time algorithms solving \mc{} when $\beta \le 3$ and~$\alpha=0$ (\Cref{theorem:d2linear,theorem:d3polynomial}).
This leads to the question whether these algorithms can be extended to matrices with arbitrary alphabet size.
Next, we are curious whether the phenomenon we observed in Theorem~\ref{theorem:aa+1} concerning the exponential dependence of the running time 
for $(\alpha, \alpha +1)$-\mc{} when $\alpha$ is even but independence 
of~$\alpha$ when it is odd can be further substantiated or whether 
one can get rid of the ``$\alpha$-dependence'' in the even case.
In terms of standard parameterized complexity analysis, we wonder whether
\mc{} is fixed-parameter tractable with respect to $\beta + k$ (in our NP-hardness proof for the case $\beta = 4$ (\Cref{theorem:d4nphard}) we have $k \in \theta(\ell)$).
Finally, performing a multivariate fine-grained 
complexity analysis in the same spirit as in recent work for 
\textsc{Longest Common Subsequence}~\cite{BK18} 
would be another natural next step.

\bibliography{ref}

\end{document}